\documentclass[12pt,a4paper]{article}

\usepackage[margin=2cm]{geometry} %[left=2cm,right=2cm,top=2cm,bottom=2cm]
\usepackage{setspace}
\usepackage{authblk}
\usepackage{amsmath}
\usepackage{amsthm}
\usepackage{amssymb}
\usepackage{mathtools}
\usepackage{mathrsfs}
\usepackage{amsfonts}
\usepackage[utf8]{inputenc}
\usepackage[T1]{fontenc}
\usepackage{graphicx}
\usepackage{enumitem}
\usepackage{color}
\usepackage[dvipsnames]{xcolor}
\usepackage{csquotes}
\usepackage[authoryear,semicolon]{natbib}
\usepackage{url}
\usepackage{ragged2e}
\usepackage[colorlinks=true,urlcolor=violet,citecolor=blue,linkcolor=red]{hyperref}
\usepackage{array}
\usepackage{bbm}
\usepackage{dsfont}
\usepackage[font=small,figurename=FIG.]{caption} %labelfont=bf,
\usepackage{subcaption}
\usepackage{soul}
\usepackage{stackrel}
\usepackage{physics}
\usepackage{tikz,lipsum,lmodern}
\usepackage[most]{tcolorbox}
\newtheorem{theorem}{Theorem}
\usepackage{threeparttable}
\usepackage{booktabs}
\usepackage{tablefootnote} % for table footnotes
\usepackage{algorithm}
\usepackage{algpseudocode}
\usepackage{doi}
\usepackage{tikz}
\usetikzlibrary{shapes.geometric, arrows.meta, positioning}

\tikzstyle{startstop} = [rectangle, rounded corners, minimum width=4cm, minimum height=1cm, text centered, draw=black, fill=white!20]
\tikzstyle{module1} = [rectangle, minimum width=5.8cm, minimum height=1.4cm, text centered, draw=black, fill=white!15]
\tikzstyle{module2} = [rectangle, minimum width=5.8cm, minimum height=1.4cm, text centered, draw=black, fill=white!15]
\tikzstyle{module3} = [rectangle, minimum width=5.8cm, minimum height=1.4cm, text centered, draw=black, fill=white!20]
\tikzstyle{module4} = [rectangle, minimum width=5.8cm, minimum height=1.4cm, text centered, draw=black, fill=white!15]
\tikzstyle{arrow} = [thick,->,>=stealth]

\providecommand{\Yij}{\boldsymbol{Y}_{ij}}
\providecommand{\Xij}{\boldsymbol{X}_{i}}
\providecommand{\tYij}{\theta_{Y_{ij}}}
\providecommand{\tXi}{\theta_{X_i}}
\providecommand{\tVi}{\theta_{V_i}}
\providecommand{\xij}{\boldsymbol{x}_{ij}}
\providecommand{\vi}{\boldsymbol{v}_{i}}

\providecommand{\epsYij}{\boldsymbol{\varepsilon}_{Y_{ij}}}
\providecommand{\epsXi}{\boldsymbol{\varepsilon}_{X_{i}}}
\providecommand{\rY}{r_{Y_{ij}}}
\providecommand{\rX}{r_{X_{i}}}

\providecommand{\xtilde}{\widetilde{\boldsymbol{x}}_{ij}}
\providecommand{\vtilde}{\widetilde{\boldsymbol{v}}_{i}}
\providecommand{\BX}{\textbf{B}^{\top} \xtilde}
\providecommand{\bi}{\boldsymbol{b}_i}

\providecommand{\YijStarStar}{\begin{pmatrix} Y_{1ij}^{*} \\ Y_{2ij}^{*} \end{pmatrix} = \rY \begin{pmatrix} \cos{(\tYij^{*})} \\ \sin{(\tYij^{*})} \end{pmatrix}}
\providecommand{\XiStarStar}{\begin{pmatrix} X_{1i}^{*} \\ X_{2i}^{*} \end{pmatrix} = \rX \begin{pmatrix} \cos{(\tXi^{*})} \\ \sin{(\tXi^{*})} \end{pmatrix}}

\providecommand{\bet}{\boldsymbol{\beta}}

\providecommand{\alp}{\boldsymbol{\alpha}}

\providecommand{\delXone}{-\delta_{X}}
\providecommand{\delXtwo}{\delta_{X}}
\providecommand{\delYone}{-\delta_{Y}}
\providecommand{\delYtwo}{\delta_{Y}}

\providecommand{\wij}{\boldsymbol{w}_{ij}}

\providecommand{\BY}{\boldsymbol{Y}}
\providecommand{\BW}{\boldsymbol{W}}

\providecommand{\BMu}{\boldsymbol{\mu}}
\providecommand{\BSig}{\boldsymbol{\Sigma}}
\providecommand{\BI}{\boldsymbol{I}}

\providecommand{\Gb}{\boldsymbol{\Gamma}_b}
\providecommand{\BG}{\mathbf{G}}
\providecommand{\BH}{\mathbf{H}}

\providecommand{\wmu}{\boldsymbol{w}^{\top} \boldsymbol{\mu}}
\providecommand{\Sb}{\boldsymbol{\Sigma}_b}

\providecommand{\PsiD}{\Psi_{\mathcal{D}}}

\DeclareMathOperator{\atantwo}{atan2}

\newcommand{\boldemdashbullet}{\mathbin{\textbf{\textemdash}\!\bullet\!\textbf{\textemdash}}}

\begin{document}
	
\title{\bf Modeling Zero-Inflated Longitudinal Circular Data Using Bayesian Methods: Application to Ophthalmology}
\author{Prajamitra Bhuyan$^\ddag$, Soutik Halder$^\dag$, Jayant Jha$^\dag$  \\ 
{\small $^\ddag$Operations Management Group, Indian Institute of Management Calcutta, Kolkata, India} \\
{\small $^\dag$Interdisciplinary Statistical Research Unit, Indian Statistical Institute, Kolkata, India}
}
\date{}
\maketitle
	
\begin{abstract}
	{This paper introduces the modeling of circular data with excess zeros under a longitudinal framework, where the response is a circular variable and the covariates can be both linear and circular in nature. In the literature, various circular-circular and circular-linear regression models have been studied and applied to different real-world problems. However, there are no models for addressing zero-inflated circular observations in the context of longitudinal studies. Motivated by a real case study, a mixed-effects two-stage model based on the projected normal distribution is proposed to handle such issues. The interpretation of the model parameters is discussed and identifiability conditions are derived. A Bayesian methodology based on Gibbs sampling technique is developed for estimating the associated model parameters. Simulation results show that the proposed method outperforms its competitors in various situations. A real dataset on post-operative astigmatism is analyzed to demonstrate the practical implementation of the proposed methodology. The use of the proposed method facilitates effective decision-making for treatment choices and in the follow-up phases.} 
    \\
    
\end{abstract}

\section{Introduction} \label{SEC:LCRM:Introduction}
In the clinical ophthalmic practice, astigmatism is among the most prevalent refractive disorders, which is caused by a deformed corneal surface or an irregular shape of the lens inside the eye \citep{Read-Collins-Carney2007}. For a healthy person, the cornea at the front of an eye has an evenly round shape that helps to concentrate the rays sharply onto the retina for clear vision. Due to astigmatism, light rays enter the anterior portion of an eye with improper refraction and have multiple focal points on the retina, affecting the quality of vision \citep{Wolffsohn-Bhogal-Shah2011}. As a result, both distant and close objects appear blurry, distorted, or fuzzy; for example, it can be difficult for an astigmatic person to read a road sign or to perform a computer task comfortably. Figure \ref{FIG:LCRM:Focus-Light-Rays} illustrates the differences between an astigmatic and healthy eye. According to the axis of astigmatism, it is often categorized as (i) with-the-rule (WTR) astigmatism, (ii) against-the-rule (ATR) astigmatism and (iii) oblique astigmatism \citep{Remon-Monsoriu-Furlan-2017}. In WTR astigmatism, the eye perceives vertical lines ($90^\circ$) more clearly than horizontal lines ($180^\circ$), but the situation is reversed under ATR astigmatism. However, oblique astigmatism arises when the steeper meridian falls within $(30^\circ, 60^\circ)$ or $(120^\circ, 150^\circ)$. Compared to WTR or ATR astigmatism, oblique astigmatism is more adverse, as it distorts the perception of horizontal and vertical objects, like letters or numbers. Figure \ref{FIG:LCRM:WTR-ATR-Oblique} displays the visual distortions due to different types of astigmatism. To make only one preferred direction, the observed angles are multiplied by 4 and then transformed by taking mod $360^\circ$.

\begin{figure}
  \centering
  \begin{minipage}{0.48\textwidth}
      \centering
      \includegraphics[width=\linewidth]{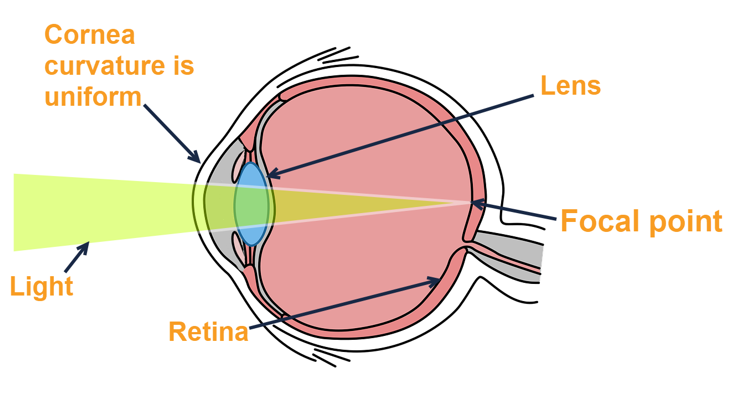}
  \end{minipage}
  \begin{minipage}{0.48\textwidth}
      \centering
      \includegraphics[width=\linewidth]{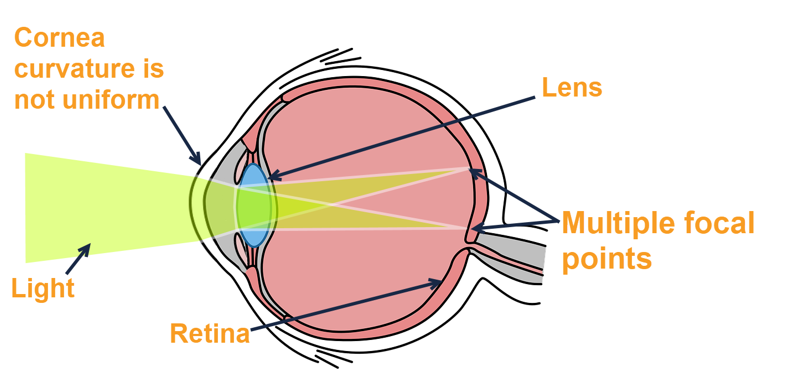}
  \end{minipage}
  \caption{Focus of light rays on the retina in a normal eye (left) and an astigmatic eye (right).}
  \label{FIG:LCRM:Focus-Light-Rays}
\end{figure}

\begin{figure}
    \centering
    \begin{subfigure}{0.85\linewidth}
        \centering
        \includegraphics[width=0.8\textwidth]{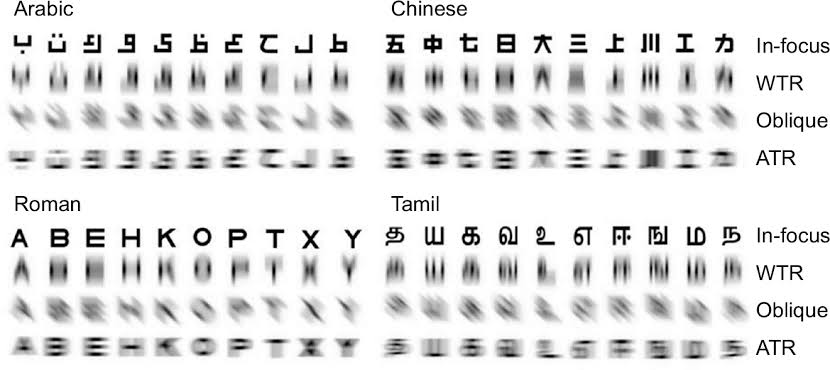}
    \end{subfigure}
    \caption{Visual distortions under WTR, ATR and Oblique astigmatism.}
    \label{FIG:LCRM:WTR-ATR-Oblique}
\end{figure}

Numerous research articles have described different causes and risk factors for astigmatism, such as hereditary, eye injuries, glaucoma, eye surgeries, etc. \citep{Shingetal2025}. According to \cite{Hashemietal2018}, the prevalence of corneal astigmatism ranges from 34.3\% to 46.6\% in adults across the World Health Organization regions and it is generally higher in the older population. Among various causes, one of the prominent reasons for getting corneal astigmatism is cataract surgery \citep{Kim-Whang-Joo2016}. The main purpose of cataract surgery is to remove the opaque natural lens from the eye and implant a new artificial intraocular lens to restore vision. Two popular methods in the field of cataract extraction treatment are Small-Incision Cataract Surgery (SICS) \citep{Singh-Misbah-Saluja-Singh2017} and PhacoEmulsification Cataract Surgery (PECS) \citep{Kelman1967}. In the former technique, a topical anesthetic is applied with a tiny self-sealing corneal incision and the cataract core is removed by creating a circular hole in the lens capsule. The latter one is an advanced technique, where the cataractous lens is emulsified with ultrasonic vibration and extracted through a microscopic incision. In underdeveloped and developing countries, the use of PECS is limited in contrast to the commonly used technique SICS due to price differences.

The statistical discipline that deals with the study of angular observations is commonly referred to as ``circular statistics" or ``directional statistics" \citep{JammalamadakaSengupta2001, MardiaJupp2009}. In contrast to the Euclidean (linear) random variables, circular variables are treated differently due to the difference in topology. Therefore, specialized statistical models are required for the analysis of circular observations. These have wide applications in many scientific disciplines, such as crystallography \citep{Mackenzie1957}, movement ecology \citep{Fisher-Lee1992}, meteorology \citep{BhattacharyaSengupta2009} and genomics \citep{SenguptaKim2016}, etc. \cite{SarmaJammalamadaka1993} suggested a general circular-circular regression model, in which the sine and cosine function of a circular response variable is regressed on a circular covariate using the trigonometric polynomial function. In contrast, \cite{Rivest1997} proposed a model for forecasting circular response by rotating the decentred circular covariate with an application to predict ground movement direction during earthquakes. With one circular response and one circular covariate, \cite{DownsMardia2002} introduced a novel circular regression approach, which was later reparameterized by \cite{Kato-Shimizu-Shieh2008}. In most applications, circular response is regressed on either a circular or linear covariate, while some articles additionally take into account multiple linear covariates and some consider only multiple circular covariates; see \cite{Fisher-Lee1992} and \cite{Jha-Biswas2017} for more details. However, various applications involve repeated measurements on circular variables. For example, in order to study the direction mechanism of birds, it is required to analyze the angular differences between the position of a bird at successive times after release \citep{Artes-Paula-Ranvaud2000}. \cite{D'EliaBorgioliScapini2001} proposed a longitudinal circular model with only linear covariates for studying the movement of small animals. A random-effects model was suggested by \cite{Rivest-Kato2019} for a very restricted setup where every covariate has a circular and a linear component. As per our knowledge, there are limited studies that regress a longitudinal circular response on circular as well as linear regressors. See \cite{Antonio-Pena2014}, \cite{Maruotti2016} for more details. However, none of these longitudinal studies consider zero-inflation in the observed circular variables, which may arise in various real-life scenarios. As discussed in \cite{Jha-Bhuyan2021} and \cite{Jha-Biswas2018}, the conventional models, which don't account for zero-inflation, induce bias in the estimates. Moreover, the zero-inflation in circular observations make the estimation computationally challenging due to intractable likelihood function. In addition, the longitudinal nature of the response variable may impose identifiability issues. We attempt to address these challenges motivated by a case study on patients affected with post-operative astigmatism.

\subsection{Challenges in analyzing astigmatism data} \label{SEC:LCRM:Challenges-Circular-Regression}
This paper considers a real case study on cataract surgery patients, treated with either SICS or PECS, observed over a period of two years (2008–2010). In this study, longitudinal measurements on the axis of astigmatism for the patients are recorded, along with their demographic profiles. See \cite{Bakshi2010} for more details about the dataset. In this paper, our primary interest is to compare the surgical procedures and the corresponding visual recovery from astigmatism after the surgery.

In this study, the measurements were taken up to the precision of $1^{\circ}$. As discussed before, the observations are multiplied by 4 and then transformed by taking mod $360^{\circ}$ to make only one preferred direction. As a result, this transformed dataset contains many zeros because of censoring of measurements on astigmatism within the range of ($-2^{\circ}, 2^{\circ}$). In the context of linear variables, some existing zero-inflated linear models are available in the literature \citep{Ghosh-Mukhopadhyay-Lu2006}. Since the topology of a circle differs from that of a line, $0^{\circ}$ cannot be regarded as a boundary point in the sample space of a circular variable. Without loss of generality, one can fix the origin of a circular variable which makes no difference in the interpretation of whether a circular variable is zero-inflated or inflated at any other angle. The zero-inflated circular modeling has been proposed by \cite{Jha-Biswas2018} and \cite{Jha-Bhuyan2021}. However, these models are not capable in handling longitudinal response and linear covariates. 
 
To address the aforementioned challenges, we introduce a mixed-effects circular regression model under a two-stage setup. The proposed model characterizes the conditional distribution of the longitudinal circular response using the projected normal (PN) distribution and its augmented density representation. This approach enables a flexible and tractable framework for modeling complex circular data, particularly accommodating various types of covariates as well as longitudinal dependence. We provide a comprehensive interpretation of the model parameters and resolve the identifiability issues arising from longitudinal dependence. The proposed model handles zero-inflation by considering a censoring mechanism motivated by the real case study under consideration. To estimate the model parameters, we develop an efficient Gibbs sampling algorithm, incorporating the identifiability conditions, within a Bayesian paradigm. We further provide the generalization of the model as well as the inferential methodology when zero-inflation occurs due to both censoring and randomness. Extensive simulations demonstrate consistent performance of the proposed model even with increasing proportions of zeros in the circular response as well as covariate. Additionally, the results from model comparison exhibit superiority of the proposed method, confirming its effectiveness and relevance in practice. A detailed sensitivity analysis further highlights the robustness and reliability of our approach. Our analysis of the astigmatism data provides assistance in comparing efficacy of the different surgical procedures. It also identifies the significant factors affecting post-surgery recovery. Finally, the results from this study provide meaningful insights and a basis for quantifying the uncertainty involved during the recovery process.

We first discuss the PN distribution for circular data in Section \ref{SEC:LCRM:Preliminaries} and describe our proposed model and methodology in Section \ref{SEC:LCRM:Proposed-Model}. In Section \ref{SEC:LCRM:Identifiability-Issues}, different identifiability issues associated with the proposed model are discussed. The estimation procedure  under a Bayesian framework is proposed in Section \ref{SEC:LCRM:Inferencial-Methodology} and a generalization is provided in Section \ref{SEC:LCRM:Some-Generalizations}. The performance of the proposed model is studied through simulation experiments in Section \ref{SEC:LCRM:Simulations}. The proposed method is applied to analyze the post-operative astigmatism data in Section \ref{SEC:LCRM:Analysis-Astigmatism-Data} and some concluding remarks are provided in Section \ref{SEC:LCRM:Discussions}.

\section{Preliminaries} \label{SEC:LCRM:Preliminaries}
A popular approach to obtain a distribution for a circular random variable is to radially project a distribution defined on $\mathbb{R}^2$ onto the unit circle. Let a random vector $\BY = (Y_1, Y_2)^\top \in \mathbb{R}^2$ follows a bivariate normal distribution, denoted by $\mathcal{N}_2(\BMu, \BSig)$, with mean vector $\BMu \in \mathbb{R}^2$ and covariance matrix $\BSig$. Then the unit vector $\BW = \frac{\BY}{R}$ follows a PN distribution, denoted by $\mathcal{PN}_2(\BMu, \BSig)$, where $R = \Vert \BY \Vert$ and $\Vert . \Vert$ denotes the Euclidean norm. Note that $\BW$ can be expressed in terms of an angular coordinate $\theta$ as $\BW = (\cos{\theta}, \sin{\theta})^{\top}$, and hence, $Y_1 = R \cos{\theta}$ and $Y_2 = R \sin{\theta}$, for $\theta \in (-\pi, \pi]$.	Further, one can obtain the joint density function of $R$ and $\theta$ as
\begin{align} \label{EQN:LCRM:Joint-Density-R,theta}
    f_{R,\theta}(r, \theta | \BMu, \BSig) = \frac{r}{2\pi \sqrt{\vert \BSig \vert}} \exp\left\{ -\frac{1}{2} (r\boldsymbol{w} - \BMu)^{\top} \BSig^{-1} (r\boldsymbol{w} - \BMu) \right\}, \quad r > 0,\ -\pi < \theta \leq \pi,
\end{align}
where $\boldsymbol{w} = (\cos{\theta}, \sin{\theta})^{\top}$. Note that $R$ is not observable and the only observed variable is $\theta$. By integrating out $R$ from (\ref{EQN:LCRM:Joint-Density-R,theta}), we can obtain the marginal density function of $\theta$, given by
\begin{align} \label{EQN:LCRM:PN-Density}
	f_{PN}(\theta | \BMu, \BSig) = \frac{1}{2\pi} |\BSig|^{-\frac{1}{2}} A_3^{-1} \exp\left\{ -\frac{1}{2} A_1 \right\} \left[ 1 + \frac{ \left( A_2 A_3^{-\frac{1}{2}} \right) \Phi \left( A_2 A_3^{-\frac{1}{2}} \right)}{\phi \left( A_2 A_3^{-\frac{1}{2}} \right)} \right],\ -\pi < \theta \leq \pi,
\end{align}
where $A_1 = \BMu^{\top} \BSig^{-1} \BMu$, $A_2 = \BMu^{\top} \BSig^{-1} \boldsymbol{w}$, $A_3 = \boldsymbol{w}^{\top} \BSig^{-1} \boldsymbol{w}$, and $\Phi(\cdot)$ and $\phi(\cdot)$ denote the cumulative distribution function and the probability density function of the standard normal distribution, respectively.
	
The mean direction of $\theta$ is defined as a unit vector $\frac{E(\BW)}{\Vert E(\BW) \Vert}$, when the mean resultant length $\Vert E(\boldsymbol{W}) \Vert\ \in (0,1]$. Note that $\Vert E(\BW) \Vert = 0$ if and only if $\BMu = \boldsymbol{0}$ and $\BSig = \sigma^2 \BI_2$ (where $\BI_2$ denotes the identity matrix of order $2$), which reduces the distribution of $\theta$ to the circular uniform distribution with undefined mean direction. When $\BSig = \BI_2$, the distribution of $\theta$ is unimodal and rotationally symmetric about its mean direction, and as $\Vert \BMu \Vert$ increases, the concentration of the density increases. In general, the PN distribution with $\BSig \neq \BI_2$ can be asymmetric and possibly bimodal \citep{WangGelfand2013}.

\subsection{Projected normal regression model} \label{SEC:LCRM:Regression-Model}
The regression model of $\theta_i$ on $x_i$ based on the PN distribution is given by
\begin{equation} \label{EQN:LCRM:PN-Regression-Model}
	\theta_i | \boldsymbol{x}_i \stackrel{ind}{\sim} \mathcal{PN}_2 \left( \left( \boldsymbol{x}_{i}^{\top} \bet_{1}, \boldsymbol{x}_{i}^{\top} \bet_{2} \right)^{\top}, \BSig \right),
\end{equation}
where $\theta_i \in (-\pi, \pi]$ is the circular response variable, $\boldsymbol{x}_i \in \mathbb{R}^p$ is the vector of linear covariates, $i = 1, 2, \ldots, n$, and $\bet_1, \bet_2 \in \mathbb{R}^p$ are the vector of regression coefficients. Assuming $\BSig = \BI_2$, \citet{Presnell1998} proposed maximum likelihood estimation method for the model parameters involved in (\ref{EQN:LCRM:PN-Regression-Model}). Under the same setup, \citet{Antonio-Pena-Escarela2011} proposed a Bayesian methodology using Metropolis-Hastings within Gibbs algorithm. 
\begin{figure}[!ht]
	\centering
	\includegraphics[width=0.5\linewidth]{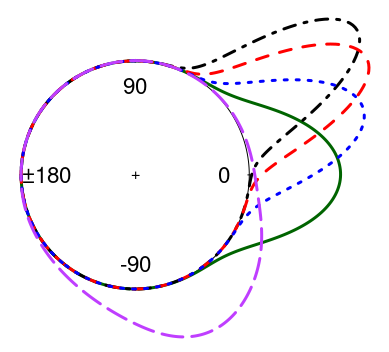}
	\caption{Circular density plots of $\mathcal{PN}_2((\beta_{10} + \beta_{11}x, \beta_{20} + \beta_{21}x)^{\top}, \BI_2)$ for $x = 0$ (dark orchid), $x = 1$ (dark green), $x = 2$ (blue), $x = 5$ (red) and $x = 100$ (black) with coefficients $\beta_{10} = 1, \beta_{11} = 3, \beta_{20} = -2$ and $\beta_{21} = 2$.} 
    \label{FIG:LCRM:Interpretation-Slopes}
\end{figure}
	
For notational simplicity, we use $\theta$ and $\boldsymbol{x}$ as the generic form of $\theta_i$ and $\boldsymbol{x}_i$, respectively. Here, we consider a linear covariate, say $x$, with $\BSig = \BI_2$,  and the angular mean direction of $\theta$ given $x$ is $m(x) = \atantwo(\beta_{20} + \beta_{21} x, \beta_{10} + \beta_{11} x)$, where `$\atantwo$' \citep[pp. 9]{AgarwalPereraPinelas2011} is defined as
\begin{equation*}
	\atantwo(S, C) = \begin{cases}
	\arctan \left( \frac{S}{C} \right) & \mbox{if $C > 0$,}  \\
	\arctan \left( \frac{S}{C} \right) + \pi \times \operatorname{sgn}(S) & \mbox{if $C < 0$,}  \\
	\frac{\pi}{2} \times \operatorname{sgn}(S) & \mbox{if $C = 0$ and $S \neq 0$,}  \\
    \text{undefined} & \mbox{if $C = S = 0$.}
	\end{cases}
\end{equation*}
Note that `$\arctan$' represents the standard inverse tangent function with range $(-\frac{\pi}{2}, \frac{\pi}{2})$, and $\operatorname{sgn}(u) = \mathds{1}(u \geq 0) - \mathds{1}(u < 0)$, for $u \in \mathbb{R}$, where $\mathds{1}(.)$ denotes the indicator function. Then, $m^{'}(x) = \frac{\beta_{10}\beta_{21} - \beta_{11}\beta_{20}}{\Vert \BMu(x) \Vert^2}$, where $m^{'}(x) = \frac{\dd}{\dd x} m(x)$ and $\Vert \BMu(x) \Vert = \sqrt{(\beta_{10} + \beta_{11}x)^2 + (\beta_{20} + \beta_{21}x)^2}$ ($\neq 0$). Thus, a higher concentration will result in a smaller shift in angular mean direction. This behaviour can also be noted in Figure \ref{FIG:LCRM:Interpretation-Slopes}. It is clear that the rate of change in the square of concentration is a linear function of $x$, and  given by
$\frac{\dd}{\dd x} \Vert \BMu(x) \Vert^2 = 2 (\beta_{10} \beta_{11} + \beta_{20} \beta_{21}) + 2 (\beta_{11}^2 + \beta_{21}^2)x.$ When $\beta_{10}\beta_{21} - \beta_{11}\beta_{20} \neq 0$, the angular mean direction of $\theta$ given $x = 0$ is $\atantwo(\beta_{20}, \beta_{10})$, and it gradually moves towards $\atantwo(\beta_{21}, \beta_{11})$ as $x$ increases. 
    
Next, we consider the case when $\beta_{10}\beta_{21} - \beta_{11}\beta_{20} = 0$. In this case, for all values of $x \in \mathbb{R}$, the angular mean direction is $\atantwo(\beta_{20}, \beta_{10})$ and the concentration is independent of $x$ provided $\beta_{11} = \beta_{21} = 0$. When $\beta_{21} > 0$, the angular mean direction is $\atantwo(\beta_{21}, \beta_{11})$ for $x > -\frac{\beta_{20}}{\beta_{21}}$, it is exactly opposite (i.e., $\atantwo(\beta_{21}, \beta_{11}) + \pi$) for $x < -\frac{\beta_{20}}{\beta_{21}}$ and is undefined at $x = -\frac{\beta_{20}}{\beta_{21}}$. For negative values of $\beta_{21}$, the roles of the two regions of $x$ will interchange. One can make similar interpretation for $\beta_{11} (\neq 0)$.

\section{Proposed model for zero-inflated longitudinal circular data} \label{SEC:LCRM:Proposed-Model}
In this section, we propose a two-stage longitudinal circular regression mixed-effects model to analyze zero-inflated longitudinal measurements on the axis of astigmatism over the successive follow-up periods. The individual patients are indexed by $i \in {I} := \{1, \ldots, n\}$ and  for each individual $i$, the measurements are taken at time points indexed by $j \in {J_i} := \{1, \ldots, m_i\}$. We denote the circular response variable, the vector of linear covariates and the circular covariate by $\tYij \in (-\pi, \pi]$, $\xij = (1, x_{1ij}, \ldots, x_{pij})^{\top} \in \mathbb{R}^{p+1}$ and $\tXi \in (-\pi, \pi]$, respectively. Now, we consider the latent variables $Y_{1ij} = \rY \cos{(\tYij)}$ and $Y_{2ij} = \rY \sin{(\tYij)}$ and define the longitudinal circular regression mixed-effects (LCRM) model as
\begin{equation} \label{EQN:LCRM:Stage-I}
    \begin{split}
        Y_{1ij} &= \beta_{10} + \sum_{k = 1}^{p} \beta_{1k} {x}_{kij} + \beta_{1C} \cos{(\tXi)} + \beta_{1S} \sin{(\tXi)} + b_{1i} + \varepsilon_{Y_{1ij}}, \\
		Y_{2ij} &= \beta_{20} + \sum_{k = 1}^{p} \beta_{2k} {x}_{kij} + \beta_{2C} \cos{(\tXi)} + \beta_{2S} \sin{(\tXi)} + b_{2i} + \varepsilon_{Y_{2ij}}.
    \end{split}
\end{equation}
In the LCRM model given in \eqref{EQN:LCRM:Stage-I}, $\xtilde = \left( 1, \xij^{\top}, \cos{(\tXi)}, \sin{(\tXi)} \right)^{\top}$ is the fixed part  and $\boldsymbol{b}_i = (b_{1i}, b_{2i})^\top$ denotes the $i$-th subject-specific random-effect that captures the longitudinal dependence. We assume that $\bi$'s are i.i.d. $\mathcal{N}_2(\boldsymbol{0}, \Sb)$ and the error vectors $\epsYij = (\varepsilon_{Y_{1ij}}, \varepsilon_{Y_{2ij}})^{\top}$ are i.i.d. $\mathcal{N}_2 (\boldsymbol{0}, \BSig_Y)$, for all $i \in I$ and $j \in J_i$. In addition, $\bi$ and $\epsYij$ are assumed to be independently distributed, for all $i \in I$, $j \in J_i$. Therefore, for a given $\boldsymbol{b}_i$, the conditional distribution of $\tYij$ is obtained as 
\begin{align*}
	\tYij | \bet_1, \bet_2, \xij, \tXi, \boldsymbol{b}_i \stackrel{ind}{\sim} \mathcal{PN}_2 \left( (\xtilde^{\top} \bet_1 + b_{1i}, \xtilde^{\top} \bet_2 + b_{2i})^{\top}, \BSig_Y \right),\ \text{for all } i \in I, j \in J_i.
\end{align*}
As discussed before, the aforementioned LCRM model is not suitable for handling excess zeros in either of the circular response and/or the circular covariate. In this context, \citet{Jha-Bhuyan2021} proposed a circular-circular regression model to account for the excess zeros, however, it could not accommodate the longitudinal responses and the linear covariates. To accommodate zero-inflation, we define the  latent variables $\tYij^*$ and $\tXi^*$ as 
\begin{equation} \label{EQN:LCRM:Latent-Theta-Y}
	\tYij =  \begin{cases}
		0 & \mbox{if $\tYij^* \in (\delYone, \delYtwo)$,}  \\
		\tYij^* & \mbox{otherwise,}
	\end{cases}
\end{equation}
and
\begin{equation} \label{EQN:LCRM:Latent-Theta-X}
	\tXi =  \begin{cases}
		0 & \mbox{if $\tXi^* \in (\delXone, \delXtwo)$,}  \\
		\tXi^* & \mbox{otherwise,}
	\end{cases}
\end{equation}
respectively, where $\delYtwo, \delXtwo \in [-\frac{\pi}{2}, \frac{\pi}{2}]$ are known constants. For modeling the zero-inflated circular covariate, we incorporate instrumental variables \citep[Ch-4, pp. 95]{CameronTrivedi2005} which may not directly affect the response $\tYij$, but it can induce changes only through the covariate $\tXi$. We denote the vector of linear instrumental variables by $\vi = (1, v_{1i}, \ldots, v_{qi})^{\top} \in \mathbb{R}^{q+1}$ and the circular instrumental variable by $\tVi \in (-\pi, \pi]$. Based on the aforementioned latent variables, we propose the two-stage LCRM model as
\begin{align}
    \begin{split} \label{EQN:LCRM:ZeroInflated-Stage-I}
        \text{Stage-I:} \quad Y_{1ij}^{*} &= \beta_{10} + \sum_{k = 1}^{p} \beta_{1k} {x}_{kij} + \beta_{1C} \cos{(\tXi^{*})} + \beta_{1S} \sin{(\tXi^{*})} + b_{1i} + \varepsilon_{Y_{1ij}}, \\
        Y_{2ij}^{*} &= \beta_{20} + \sum_{k = 1}^{p} \beta_{2k} {x}_{kij} + \beta_{2C} \cos{(\tXi^{*})} + \beta_{2S} \sin{(\tXi^{*})} + b_{2i} + \varepsilon_{Y_{2ij}},
    \end{split} \\
    \begin{split} \label{EQN:LCRM:ZeroInflated-Stage-II}
        \text{Stage-II:} \quad X_{1i}^{*} &= \alpha_{10} + \sum_{k = 1}^{q} \alpha_{1k} {v}_{ki} + \alpha_{1C} \cos{(\tVi)} + \alpha_{1S} \sin{(\tVi)} + \varepsilon_{X_{1i}}, \\
        X_{2i}^{*} &= \alpha_{20} + \sum_{k = 1}^{q} \alpha_{2k} {v}_{ki} + \alpha_{2C} \cos{(\tVi)} + \alpha_{2S} \sin{(\tVi)} + \varepsilon_{X_{2i}},
    \end{split}
\end{align}
where $$\YijStarStar \quad \text{and} \quad \XiStarStar,$$
for $i \in I, j \in J_i.$ The distributional assumptions for $\bi$ and $\epsYij$ are the same as those considered before. Additionally, we assume $\epsXi = (\varepsilon_{X_{1i}}, \varepsilon_{X_{2i}})^{\top} \stackrel{iid}{\sim} \mathcal{N}_2 (\boldsymbol{0}, \BSig_X)$, and are independent of $\epsYij$ and $\boldsymbol{b}_i$, for all $i\in I$, $j\in J_{i}$. Notably, if $\delYtwo = 0$, the Stage-I of the aforementioned model reduces to the LCRM model, given in (\ref{EQN:LCRM:Stage-I}). Similarly, the Stage-II model reduces to the PN regression model, given in (\ref{EQN:LCRM:PN-Regression-Model}), when $\delXtwo = 0$.

\section{Identifiability issues} \label{SEC:LCRM:Identifiability-Issues}
\citet{Presnell1998} discussed the identifiability issue for the PN distribution, given in \eqref{EQN:LCRM:PN-Density}, and proposed to constrain the distribution by considering $\BSig = \BI_2$. Following the same proposal, we take $\BSig_X = \BSig_Y=\BI_2$ in the two-stage LCRM model proposed in Section \ref{SEC:LCRM:Proposed-Model}. However, these restrictions are not sufficient to ensure the identifiability of the LCRM model, given in \eqref{EQN:LCRM:ZeroInflated-Stage-I}, due to the unconstrained structure of the covariance matrix $\Sb$ of the random-effects distribution. Note that the conditional distribution of $\Yij = (Y_{1ij}, Y_{2ij})^{\top}$ given $\boldsymbol{b}_i$ is $\mathcal{N}_2 \left( \BX + \boldsymbol{b}_i, \BI_2 \right)$ and the marginal distribution of $\Yij$ is $\mathcal{N}_2 \left( \BX, \Gb \right)$, where $\BX = \left(\xtilde^{\top} \bet_1, \xtilde^{\top} \bet_2 \right)^{\top}$ and $\Gb = \BI_2 + \Sb$. Therefore, the circular response $\tYij \sim \mathcal{PN}_2 \left(\BX, \Gb \right)$. For any $c > 1$, consider the following identity given as

\begin{align} \label{EQN:LCRM:Idenfiability-2}
	&f_{PN}(\tYij | c\BX, c^2\Gb) \notag \\
	&= \frac{1}{2\pi} {|c^2\Gb|}^{-\frac{1}{2}} A_{3ij}(c)^{-1} \exp \left\{-\frac{1}{2} A_{1ij}(c) \right\} \left[1 + \frac{ \left( A_{2ij}(c) A_{3ij}(c)^{-\frac{1}{2}} \right)\ \Phi \left( A_{2ij}(c) A_{3ij}(c)^{-\frac{1}{2}} \right)}{\phi \left( A_{2ij}(c) A_{3ij}(c)^{-\frac{1}{2}} \right)} \right] \notag \\
	&= \frac{1}{2\pi} |\Gb|^{-\frac{1}{2}} A_{3ij}^{-1} \exp \left\{-\frac{1}{2} A_{1ij} \right\} \left[1 + \frac{ \left( A_{2ij} A_{3ij}^{-\frac{1}{2}} \right)\ \Phi \left( A_{2ij} A_{3ij}^{-\frac{1}{2}} \right)}{\phi \left( A_{2ij} A_{3ij}^{-\frac{1}{2}} \right)} \right] \notag \\
	&= f_{PN} (\tYij | \BX, \Gb),
\end{align}
where ${A_{1ij}(c) = (c\BX)^{\top} (c^2\Gb)^{-1} (c\BX) = A_{1ij}}$, ${A_{2ij}(c) = (c\BX)^\top (c^2\Gb)^{-1} \wij = c^{-1} A_{2ij}}$, and ${A_{3ij}(c) = \wij^{\top} (c^2\Gb)^{-1} \wij = c^{-2} A_{3ij}}$. Hence, as shown in (\ref{EQN:LCRM:Idenfiability-2}), the marginal density of $\tYij$ remains unaltered even if $(\textbf{B}, \Gb) \neq (c\textbf{B}, c^2\Gb)$ for $c>1$. Thus, the proposed model fails to be identifiable. As per our knowledge, identifiability issues for LCRM model have not been discussed in the literature. To resolve this issue, we provide a sufficient condition for the identifiability in the following theorem. 
\begin{theorem} \label{THEOREM:LCRM:Theorem-1}
    If $\tYij \sim \mathcal{PN}_2 (\BMu_{Y_{ij}}, \Gb)$, where $\BMu_{Y_{ij}} = \BX$ and $\Gb = \BI_2 + \Sb$, then the model, given by (\ref{EQN:LCRM:Stage-I}), is identifiable if the generalized variance of the random-effect $\boldsymbol{b}_i$ is equal to 1 (i.e., $|\Sb| = 1$). 
\end{theorem}
\begin{proof}
A detailed proof of the theorem is provided in Section A of the Supplementary File.
\end{proof}

Note that the diagonal elements of the covariance matrix $\Sb$ incorporate the longitudinal dependence within $Y_{1ij}$ and $Y_{2ij}$, respectively, whereas, the off-diagonal element of $\Sb$ captures the dependence between $Y_{1ij}$ and $Y_{2ij}$. Since $r_{Yij}$ is not observable, the aforementioned approach induces redundancy in the model. Hence, the restriction on the generalized variance of the random-effects distribution makes the LCRM model identifiable.

\section{Inferential methodology} \label{SEC:LCRM:Inferencial-Methodology}
We propose a Bayesian methodology using the Gibbs sampler for estimating the model parameters involved in (\ref{EQN:LCRM:ZeroInflated-Stage-I}) and (\ref{EQN:LCRM:ZeroInflated-Stage-II}). Since the likelihood function is analytically intractable, we use data augmentation to work with the complete-data likelihood function. This approach helps to construct the full conditionals and perform the Gibbs sampling algorithm.

\subsection{Bayesian estimation} \label{SEC:LCRM:Bayes-estimation}
Let us denote the latent radii of the circular response and the circular covariate by $\Psi_{r_Y} = \{\rY: i \in I, j \in J_i\}$ and $\Psi_{r_X} = \{\rX: i \in I\}$, respectively. Also, denote the circular latent variables mentioned in (\ref{EQN:LCRM:Latent-Theta-Y}) and (\ref{EQN:LCRM:Latent-Theta-X}) by $\Psi_{\theta_Y} = \{\tYij^*: i \in I, j \in J_i \}$ and $\Psi_{\theta_X} = \{\tXi^*: i \in I \}$, respectively, and $\Psi_b = \{\boldsymbol{b}_{i}: i \in I \}$ denotes the random-effects described in (\ref{EQN:LCRM:ZeroInflated-Stage-I}). Let us further represent the latent vectors by $\Psi_Y = \{\Yij^{*}: i \in I, j \in J_i\}$ and $\Psi_X = \{\Xij^{*}: i \in I\}$, and the observed data by $\mathcal{D} = \{(\tYij, \xij, \tXi, \vi, \tVi): i \in I, j \in J_i \}$. Also, define $\PsiD = \{\Psi_{r_Y}, \Psi_{\theta_Y}, \Psi_{r_X}, \Psi_{\theta_X}, \Psi_b, \mathcal{D} \}$ and $\PsiD^{*} = \{\Psi_{Y}, \Psi_{X}, \Psi_b, \mathcal{D} \}$.
	
We consider conjugate prior densities on $\bet_k$ and $\alp_k$ as $\pi(\bet_k | \BMu_{\bet_k}, \BSig_{\bet_k}) \equiv \mathcal{N}_{p+3} (\BMu_{\bet_k}, \BSig_{\bet_k})$ and $\pi(\alp_k | \BMu_{\alp_k}, \BSig_{\alp_k}) \equiv \mathcal{N}_{q+3} (\BMu_{\alp_k}, \BSig_{\alp_k})$, respectively, for $k = 1,2$. In order to enforce identifiability of the model parameters and conjugacy of the associated posteriors, we reparameterize $\Sb = \left (\begin{smallmatrix}\sigma_1^2 & \rho \sigma_1 \sigma_2 \\ \rho \sigma_1 \sigma_2 & \sigma_2^2\end{smallmatrix} \right)$ as $s_1 = \rho \frac{\sigma_2}{\sigma_1}$ and $s_2 = \sigma_2^2 (1-\rho^2)$. In this setup, we consider a normal-gamma prior: $\pi(s_1 | \tau, \lambda_0) \equiv \mathcal{N} (0, \frac{1}{\tau\lambda_{0}})$ and $\pi(\tau | \nu_0, \kappa_0) \equiv \mathcal{G} (\nu_0, \kappa_0)$, where $\tau = 1/s_2$, and $\lambda_0$, $\nu_0$ and $\kappa_0$ are known hyperparameters. We denote the set of unknown parameters by $\Delta = \{ \bet_1, \bet_2, s_1, \tau, \alp_1, \alp_2 \}$, and $\Delta(-\chi)$ denotes the whole set $\Delta$ except the unknown entity $\chi$. The joint posterior density of the unknown model parameters $\Delta$, the random-effect $\boldsymbol{b}_i$ and the latent variables $\Psi_Y$ and $\Psi_X$, involved in (\ref{EQN:LCRM:ZeroInflated-Stage-I}) and (\ref{EQN:LCRM:ZeroInflated-Stage-II}), can be written as
\begin{align*}
    \pi &(\Delta, \Psi_{b}, \Psi_{Y}, \Psi_{X} | \mathcal{D}) \propto \pi(\bet_1 | \BMu_{\bet_1}, \BSig_{\bet_1})\ \pi(\bet_2 | \BMu_{\bet_2}, \BSig_{\bet_2})\ \pi(s_1 | \tau, \lambda_0)\ \pi(\tau | \nu_0, \kappa_0)\ \pi(\alp_1 | \BMu_{\alp_1}, \BSig_{\alp_1}) \\
    &\times \pi(\alp_2 | \BMu_{\alp_2}, \BSig_{\alp_2}) \times \prod_{i=1}^n \Bigg[ f_{N} \left(\Xij^{*} \vert \BMu_{X_{i}}, \BI_2 \right) \times f_{N} \left(\boldsymbol{b}_{i} \vert \boldsymbol{0}, \Sb \right) \times \Big\{ \prod_{j=1}^{m_i} f_{N} \left(\Yij^{*} \vert \BMu_{Y_{ij}}^{*} + \boldsymbol{b}_i, \BI_2 \right) \Big\} \Bigg],
\end{align*}
where $\BMu_{X_{i}} = (\vtilde^{\top} \alp_1, \vtilde^{\top} \alp_2)^{\top}$, $\vtilde = (1, \vi^{\top}, \cos{(\tVi)}, \sin{(\tVi)})^{\top}$, $\BMu_{Y_{ij}}^{*} = (\xtilde^{*\top} \bet_1, \xtilde^{*\top} \bet_2)^{\top}$ and $\xtilde^{*} = (1, \xij^{\top}, \cos{(\tXi^{*})}, \sin{(\tXi^{*})})^{\top}$, and $f_{N}(. \vert \BMu, \BSig)$ denotes the bivariate normal density with parameters $\BMu$ and $\BSig$. Alternatively, one can consider the latent variables ($\Psi_{r_Y}, \Psi_{\theta_Y}$, $\Psi_{r_X}, \Psi_{\theta_X}$), instead of $\Psi_Y$ and $\Psi_X$, and obtain the joint posterior density as

\begin{align*}
    \pi &(\Delta, \Psi_b, \Psi_{r_Y}, \Psi_{\theta_Y}, \Psi_{r_X}, \Psi_{\theta_X} | \mathcal{D})
    \propto \pi(\bet_1 | \BMu_{\bet_1}, \BSig_{\bet_1})\ \pi(\bet_2 | \BMu_{\bet_2}, \BSig_{\bet_2})\ \pi(s_1 | \tau, \lambda_0)\ \pi(\tau | \nu_0, \kappa_0) \\
    &\times \pi(\alp_1 | \BMu_{\alp_1}, \BSig_{\alp_1})\ \pi(\alp_2 | \BMu_{\alp_2}, \BSig_{\alp_2}) \times \prod_{i=1}^{n} \Bigg[ \Big\{ \left[ f_{R}(\rX | \tXi, \BMu_{X_{i}}) \times f_{PN}(\tXi | \BMu_{X_{i}}, \boldsymbol{I}_2) \right]^{\mathds{1}(\tXi \neq 0)} \\
    &\times \left[ f_{R}(\rX | \tXi^{*}, \BMu_{X_{i}}) \times C_{X_{i}} \times f_{TPN}(\tXi^{*} | \BMu_{X_{i}}, \boldsymbol{I}_2; \delXone, \delXtwo) \right]^{\mathds{1}(\tXi = 0)} \Big\} \Bigg] \\
    &\times \prod_{i=1}^n \Bigg[ f_{N} \left(\boldsymbol{b}_{i} | \boldsymbol{0}, \Sb \right) \times \prod_{j=1}^{m_i} \Big\{ \left[ f_{R}(\rY | \tYij, \BMu_{Y_{ij}}^{*} + \boldsymbol{b}_{i}) \times f_{PN}(\tYij | \BMu_{Y_{ij}}^{*} + \boldsymbol{b}_{i}, \boldsymbol{I}_2) \right]^{\mathds{1}(\tYij \neq 0)} \\
    &\times \left[ f_{R}(\rY | \tYij^{*}, \BMu_{Y_{ij}}^{*} + \boldsymbol{b}_{i}) \times C_{Y_{ij}} \times f_{TPN}(\tYij^{*} | \BMu_{Y_{ij}}^{*} + \boldsymbol{b}_{i}, \boldsymbol{I}_2; \delYone, \delYtwo) \right]^{\mathds{1}(\tYij = 0)} \Big\} \Bigg],
\end{align*}
where $C_{Y_{ij}} = \int_{\delYone}^{\delYtwo} f_{PN}(\tYij^{*} | \BMu_{Y_{ij}}^{*} + \boldsymbol{b}_{i}, \BI_2) d\tYij^{*}$ and $C_{X_{i}} = \int_{\delXone}^{\delXtwo} f_{PN}(\tXi^{*} | \BMu_{X_{i}}, \BI_2) d\tXi^{*}$. Here, the density of the truncated projected normal (TPN) distribution is denoted by $f_{TPN}(\theta^{*} | \BMu, \BI_2; -\delta, \delta)$ $= C^{-1} f_{PN}(\theta^{*} | \BMu, \BI_2) \mathds{1}(\theta^{*} \in (-\delta, \delta))$ with $C = \int_{-\delta}^{\delta} f_{PN}(\theta^{*} | \BMu, \BI_2) d\theta^{*}$. Now we obtain the full conditional densities of the latent variables $\tYij^*$, $\rY$, $\tXi^*$ and $\rX$, for all $i \in I, j \in J_i$, the random-effect $\boldsymbol{b}_{i}$, for each $i \in I$, and the model parameters involved in (\ref{EQN:LCRM:ZeroInflated-Stage-I}) and (\ref{EQN:LCRM:ZeroInflated-Stage-II}), from the above joint posterior densities. These conditional densities have the following form:

\begin{align}
	\pi (\tYij^{*} | \Delta, \PsiD(-\Psi_{\theta_Y})) &\equiv 
	\begin{cases}
		\mathds{1}(\tYij^{*} = \tYij) & \mbox{if $\tYij \neq 0,$} \\
		f_{TPN}(\tYij^{*} | \BMu_{Y_{ij}}^{*} + \boldsymbol{b}_{i}, \boldsymbol{I}_2; \delYone, \delYtwo) & \mbox{otherwise},
	\end{cases}  \label{EQN:LCRM:Conditional-thetaY} \\[1em]
	\pi (\rY | \Delta, \PsiD(-\Psi_{r_Y})) &\equiv K_{Y}^{-1} \rY \exp \left\{-\frac{1}{2} \left(\rY - a_{Y_{ij}}^{*} \right)^2 \right\} \mathds{1}(\rY \in (0, \infty)),  \label{EQN:LCRM:Conditional-rY} \\[1em]
	\pi (\bet_k | \Delta(-\bet_k), \PsiD^{*}) &\equiv \mathcal{N}_{p+3} \left( \BG_k^{-1} \Big\{ \sum_{i=1}^n \sum_{j=1}^{m_i} \xtilde^{*} (Y_{kij}^{*} - b_{ki}) + \BSig_{\bet_k}^{-1} \BMu_{\bet_k} \Big\}, \BG_k^{-1} \right),  \label{EQN:LCRM:Conditional-betak-conjugate} \\[1em]
	\pi (\boldsymbol{b}_{i} | \Delta, \PsiD^{*}(-\Psi_b)) &\equiv \mathcal{N}_2 \left(\boldsymbol{S}_b^{-1} \sum_{j=1}^{m_i} \left(\Yij^{*} - \BX^{*}\right), \boldsymbol{S}_b^{-1}\right),\ \text{for each } i \in I,  \label{EQN:LCRM:Conditional-bi-conjugate}, \\[1em]
    \pi (s_1 | \Delta(-s_1), \PsiD^{*}) &\equiv \mathcal{N} \left(\frac{\sum_{i=1}^n b_{1i} b_{2i}}{\lambda_0 + \sum_{i=1}^n b_{1i}^2}, \frac{1}{\tau (\lambda_0 + \sum_{i=1}^n b_{1i}^2)} \right),   \label{EQN:LCRM:Conditional-s1-conjugate} \\[1em]
	\pi (\tau | \Delta(-\tau), \PsiD^{*}) &\equiv \mathcal{G} \left(\nu_0 + \frac{n+1}{2}, \kappa_0 + \frac{\sum_{i=1}^n (b_{2i} - s_1 b_{1i})^2 + \lambda_0 s_1^2}{2} \right),   \label{EQN:LCRM:Conditional-tau-conjugate} \\[1em]
    \pi (\tXi^{*} | \Delta, \PsiD(-\Psi_{\theta_X})) &\equiv 
	\begin{cases}
		\mathds{1}(\tXi^{*} = \tXi) & \mbox{if $\tXi \neq 0,$} \\
		f_{TPN}(\tXi^{*} | \BMu_{X_{i}}, \boldsymbol{I}_2; \delXone, \delXtwo) & \mbox{otherwise},
	\end{cases}   \label{EQN:LCRM:Conditional-thetaX} \\[1em]
	\pi (\rX | \Delta, \PsiD(-\Psi_{r_X})) &\equiv K_{X}^{-1} \rX \exp \left\{-\frac{1}{2} \left(\rX - a_{X_{i}}^{*} \right)^2 \right\} \mathds{1}(\rX \in (0, \infty)),  \label{EQN:LCRM:Conditional-rX} \\[1em]
	\pi (\alp_k | \Delta(-\alp_k), \PsiD^{*}) &\equiv \mathcal{N}_{q+3} \left( \BH_k^{-1} \Big\{ \sum_{i=1}^n \vtilde X_{ki}^{*} + \BSig_{\alp_k}^{-1} \BMu_{\alp_k} \Big\}, \BH_k^{-1} \right),  \label{EQN:LCRM:Conditional-alphak-conjugate}
\end{align}
for $k = 1,2$, where $K_{Y} = \sqrt{2\pi} [\phi(a_{Y_{ij}}^{*}) + a_{Y_{ij}}^{*} \Phi(a_{Y_{ij}}^{*})]$, $a_{Y_{ij}}^{*} = \boldsymbol{w}_{Y_{ij}}^{*\top} (\BMu_{Y_{ij}}^{*} + \boldsymbol{b}_{i})$, \\ $\boldsymbol{w}_{Y_{ij}}^{*} = (\cos{(\tYij^{*})}, \sin{(\tYij^{*})})^{\top}$, $\BG_k = \sum_{i=1}^n \sum_{j=1}^{m_i} \xtilde^{*} \xtilde^{*\top} + \BSig_{\bet_k}^{-1}$, $\boldsymbol{S}_b = (m_i\BI_2 + \Sb^{-1})$, $K_{X} = \sqrt{2\pi} [\phi(a_{X_{i}}^{*}) + a_{X_{i}}^{*} \Phi(a_{X_{i}}^{*})]$, $a_{X_{i}}^{*} = \boldsymbol{w}_{X_{i}}^{*\top} \BMu_{X_{i}}$, $\boldsymbol{w}_{X_{i}}^{*} = (\cos{(\tXi^{*})}, \sin{(\tXi^{*})})^{\top}$ and $\BH_k = \sum_{i=1}^n \vtilde \vtilde^{\top} + \BSig_{\alp_k}^{-1}$. After drawing the posterior samples of $s_1$ and $s_2 \left(= 1/\tau \right)$ from the corresponding full conditional densities given in (\ref{EQN:LCRM:Conditional-s1-conjugate}) and (\ref{EQN:LCRM:Conditional-tau-conjugate}), we obtain the posterior samples of $\rho$ and $\sigma_2^2$. Note that the samples from the conditional posterior distribution of $\sigma_1^2$ are generated using the identifiability constraint $\sigma_1^2 = 1/(\sigma_2^2 (1 - \rho^2))$. The derivations of the full conditional densities of the associated model parameters and random-effects are provided in Section B of the Supplementary File. The sampling algorithms for generating latent variables, mentioned in (\ref{EQN:LCRM:Conditional-thetaY}), (\ref{EQN:LCRM:Conditional-rY}), (\ref{EQN:LCRM:Conditional-thetaX}) and (\ref{EQN:LCRM:Conditional-rX}), are discussed in the following subsection.

\subsection{Sampling procedures} \label{SEC:LCRM:Sampling-algorithms}
One can consider accept-reject algorithm to generate samples from the conditional distribution of $\tYij^{*}$ when $\tYij=0$. In that case, one can draw observations from $f_{PN}(\tYij^{*} | \BMu_{Y_{ij}}^{*} + \boldsymbol{b}_{i}, \boldsymbol{I}_2)$ and then accept those which lie within $(\delYone, \delYtwo)$. Similarly, when $\tXi=0$, $\tXi^{*}$  can be drawn from the interval $(\delXone, \delXtwo)$. However, the acceptance rate is very low for small values of $\delYtwo$ or $\delXtwo$. To address this issue, we propose a generic algorithm to generate samples from $f_{TPN}(\theta^{*} | \boldsymbol{\mu}, \boldsymbol{I}_2; \delta_1, \delta_2)$, where $-\pi < \delta_1 < \delta_2 \leq \pi$. Note that for our application $\delta_1=-\delta_2$. We denote the truncated normal density by $f_{TN}(. \vert \mu, \sigma^{2}; a_{1}, a_{2})$ with parameters $\mu$ and $\sigma^{2}$ and the truncation range $(a_{1}, a_{2})$. Also, we define ${I(k_{1}, k_{2}) \coloneqq \int_{k_{1}}^{k_{2}} f_{PN}(\theta^{*} \mid \BMu, \BI_2) d\theta^{*}}$, where ${-\pi < k_{1} \leq k_{2} \leq \pi}$. The proposed algorithm is provided below.
\begin{algorithm}[ht!]
\caption*{\textbf{Sampling algorithm from a TPN distribution}}
\small
\begin{algorithmic}
    \State sample $z_{1} \sim f_{TN}(z_{1} \mid \mu_{1}, 1; 0, \infty)$;

    \If{$\delta_{1}, \delta_{2} \in [-\frac{\pi}{2}, \frac{\pi}{2}]$,}
        \State sample $z_{2} \mid z_{1} \sim f_{TN}(z_{2} \mid \mu_{2}, 1; z_{1} \tan{\delta_{1}}, z_{1} \tan{\delta_{2}})$;
        \State calculate $\theta^{*} = \atantwo(z_{2}, z_{1})$;

    \ElsIf{$\delta_{1}, \delta_{2} \in (\frac{\pi}{2}, \pi]$,}
        \State sample $z_{2} \mid z_{1} \sim f_{TN}(z_{2} \mid \mu_{2}, 1; z_{1} \tan{\delta_{1}}, z_{1} \tan{\delta_{2}})$;
        \State calculate $\theta^{*} = \atantwo(z_{2}, z_{1}) + \pi$;

    \ElsIf{$\delta_{1}, \delta_{2} \in (-\pi, -\frac{\pi}{2})$,}
        \State sample $z_{2} \mid z_{1} \sim f_{TN}(z_{2} \mid \mu_{2}, 1; z_{1} \tan{\delta_{1}}, z_{1} \tan{\delta_{2}})$;
        \State calculate $\theta^{*} = \atantwo(z_{2}, z_{1}) - \pi$;

    \ElsIf{$\delta_{1} \in (-\pi, -\frac{\pi}{2})$ and $\delta_{2} \in [-\frac{\pi}{2}, \frac{\pi}{2}]$,}
        \State sample $H_{1} \sim \text{Bernoulli}(p_{1})$ with $p_{1} = I(\delta_{1}, -\frac{\pi}{2})/I(\delta_{1}, \delta_{2})$;
        \State sample ${z_{2} \mid z_{1}, H_{1} \sim p_{1} f_{TN}(z_{2} \mid \mu_{2}, 1; z_{1} \tan{\delta_{1}}, \infty) + (1 - p_{1}) f_{TN}(z_{2} \mid \mu_{2}, 1; -\infty, z_{1} \tan{\delta_{2}})}$; 
        \State calculate $\theta^{*} = \left(\atantwo(z_{2}, z_{1}) - \pi \right) \mathds{1}(H_{1} = 1) + \atantwo(z_{2}, z_{1}) \mathds{1}(H_{1} = 0)$;

    \ElsIf{$\delta_{1} \in [-\frac{\pi}{2}, \frac{\pi}{2}]$ and $\delta_{2} \in (\frac{\pi}{2}, \pi]$,}
        \State sample $H_{2} \sim \text{Bernoulli}(p_{2})$ with $p_{2} = I(\delta_{1}, \frac{\pi}{2})/I(\delta_{1}, \delta_{2})$;
        \State sample ${z_{2} \mid z_{1}, H_{2} \sim p_{2} f_{TN}(z_{2} \mid \mu_{2}, 1; z_{1} \tan{\delta_{1}}, \infty) + (1 - p_{2}) f_{TN}(z_{2} \mid \mu_{2}, 1; -\infty, z_{1} \tan{\delta_{2}})}$; 
        \State calculate $\theta^{*} = \atantwo(z_{2}, z_{1}) \mathds{1}(H_{2} = 1) + \left( \atantwo(z_{2}, z_{1}) + \pi \right) \mathds{1}(H_{2} = 0)$;

    \ElsIf{$\delta_{1} \in (-\pi, -\frac{\pi}{2})$ and $\delta_{2} \in (\frac{\pi}{2}, \pi]$,}
        \State sample $(H_{1}, H_{2}) \sim \text{Trinomial}(1, p_{1}, p_{2})$ with ${p_{1} = I(\delta_{1}, -\frac{\pi}{2})/I(\delta_{1}, \delta_{2}), p_{2} = I(-\frac{\pi}{2}, \frac{\pi}{2})/I(\delta_{1}, \delta_{2})}$;
        % \State \hspace{3.4cm} and $p_{2} = I(-\frac{\pi}{2}, \frac{\pi}{2})/I(\delta_{1}, \delta_{2})$;
        \State sample ${z_{2} \mid z_{1}, H_{1}, H_{2} \sim p_{1} f_{TN}(z_{2} \mid \mu_{2}, 1; z_{1} \tan{\delta_{1}}, \infty) + p_{2} \phi(z_{2} - \mu_{2})}$
        \State \hspace{4.2cm} ${+ (1 - p_{1} - p_{2}) f_{TN}(z_{2} \mid \mu_{2}, 1; -\infty, z_{1} \tan{\delta_{2}})}$;
        \State calculate $\theta^{*} = {\left(\atantwo(z_{2}, z_{1}) - \pi \right) \mathds{1}(H_{1} = 1, H_{2} = 0) + \atantwo(z_{2}, z_{1}) \mathds{1}(H_{1} = 0, H_{2} = 1)}$ 
        \State \hspace{7.6cm} ${+ \left( \atantwo(z_{2}, z_{1}) + \pi \right) \mathds{1}(H_{1} = 0, H_{2} = 0)}$;   
    \EndIf
\end{algorithmic}
\end{algorithm}

For the purpose of generating the latent variables $\rY$ and $\rX$ using (\ref{EQN:LCRM:Conditional-rY}) and (\ref{EQN:LCRM:Conditional-rX}), respectively, we apply the sampling algorithm proposed by \citet{Hernandez2017} given below. 
\begin{algorithm}[ht!] 
	\caption*{\textbf{Sampling algorithm for generating $R$}}
    \small
	\begin{algorithmic}
		\State take an initial value $R = r_0$;
		\State sample $t_0 \sim \mathcal{U} \Big(0, \exp \Big\{-\frac{1}{2} (r_0 - \wmu)^2 \Big\} \Big)$;
		\State calculate $\zeta_1 = \max \Big\{0, \wmu - \sqrt{-2\log(t_0)} \Big\}$ and $\zeta_2 = \wmu + \sqrt{-2\log(t_0)}$ (using $t_0$);
		\State draw $\kappa \sim \mathcal{U}(0,1)$ independently;
		\State update $r = \sqrt{(\zeta_2^2 - \zeta_1^2) \kappa + \zeta_1^2}$ (using the inverse cumulative distribution function technique)
	\end{algorithmic}
\end{algorithm}

\section{Generalization} \label{SEC:LCRM:Some-Generalizations}
In Section \ref{SEC:LCRM:Proposed-Model}, we consider that zero observations occur because of censoring within a specified interval. However, in some real scenarios, the presence of zeros may be due to two reasons: randomness as well as censoring. To model the aforementioned phenomenon, one can generalize the proposed model, given by (\ref{EQN:LCRM:ZeroInflated-Stage-I}) and (\ref{EQN:LCRM:ZeroInflated-Stage-II}).  We denote the latent indicator variable $Z_{Y_{ij}}$ ($Z_{X_i}$), which takes value 1 if zero occurs due to randomness, and 0, otherwise. Then, we can express $\tYij$ and $\tXi$ as follows:
\begin{align} 
    \tYij &= Z_{Y_{ij}} \delta_{0}^{Y} + (1 - Z_{Y_{ij}})\ \tYij^{*} \mathds{1}(\tYij^* \notin (\delYone, \delYtwo)), \nonumber \\
    \tXi &= Z_{X_{i}} \delta_{0}^{X} + (1 - Z_{X_{i}})\ \tXi^{*} \mathds{1}(\tXi^* \notin (\delXone, \delXtwo)), \nonumber
\end{align}
where $\delta_{0}^{Y}$ ($\delta_{0}^{X}$) represents a degenerate distribution where the probability mass function is $1$ at $\tYij = 0$ ($\tXi = 0$), and $0$, otherwise. Note that the latent variables $\tYij^*$ and $\tXi^*$ are the same as given in (\ref{EQN:LCRM:Latent-Theta-Y}) and (\ref{EQN:LCRM:Latent-Theta-X}), respectively. Here, $Z_{Y_{ij}}$ and $Z_{X_i}$ are Bernoulli random variables with parameters $\eta_Y$ and $\eta_X$, respectively. As discussed in Section \ref{SEC:LCRM:Bayes-estimation}, the same set of conjugate priors on $\bet_1$, $\bet_2$, $s_1$, $\tau$, $\alp_1$ and $\alp_2$ are considered for the aforementioned model. We also consider conjugate priors for $\eta_{Y}$ and $\eta_{X}$ as $\pi(\eta_{Y} \vert c_{Y}, d_{Y}) \equiv \text{Beta}(c_{Y}, d_{Y})$ and $\pi(\eta_{X} \vert c_{X}, d_{X}) \equiv \text{Beta}(c_{X}, d_{X})$, respectively, where $c_{Y}$, $d_{Y}$, $c_{X}$ and $d_{X}$ are known hyperparameters. The full conditional densities of the parameters and latent variables involved in the aforementioned model, along with the Gibbs sampler, are provided in Section C of the Supplementary File.

\section{Simulations} \label{SEC:LCRM:Simulations}
We study the performance of our proposed method through extensive simulations. For this purpose, we consider different choices of parameter values and generate data for three different sample sizes, $n = 50, 100,$ and $500$. In Stage-I, one linear covariate $x_{1ij}$ and a circular covariate $\tXi$ are considered, and three repeated measurements are generated for each individual. In Stage-II, we use one linear instrumental variable $v_{1i}$. The simulated data are obtained using the following specifications of the zero-inflated LCRM model in a two-stage setup as
\begin{align*}
	\begin{pmatrix} Y_{1ij}^{*} \\ Y_{2ij}^{*} \end{pmatrix} \Big| \bet_1, \bet_2, x_{1ij}, \tXi^{*}, \boldsymbol{b}_{i} &\sim \mathcal{N}_2 \left( \begin{pmatrix} \xtilde^{*\top} \bet_1 + b_{1i} \\ \xtilde^{*\top} \bet_2 + b_{2i} \end{pmatrix}, \BI_2 \right), \quad \begin{pmatrix} b_{1i} \\ b_{2i} \end{pmatrix} \Big| \Sb \sim \mathcal{N}_2 \left(\boldsymbol{0}, \Sb \right), \\ 
	\text{and} \quad 
	\begin{pmatrix} X_{1i}^{*} \\ X_{2i}^{*} \end{pmatrix} \Big| \alp_1, \alp_2, v_{1i} &\sim \mathcal{N}_2 \left( \begin{pmatrix} \vtilde^{\top} \alp_1 \\ \vtilde^{\top} \alp_2 \end{pmatrix}, \BI_2 \right),\ i = 1, 2,\ldots, n,\ j = 1, 2, 3,
\end{align*}
where $\bet_1 = (\beta_{10}, \beta_{11}, \beta_{1C}, \beta_{1S})^{\top}$, $\bet_2 = (\beta_{20}, \beta_{21}, \beta_{2C}, \beta_{2S})^{\top}$, $\alp_1 = (\alpha_{10}, \alpha_{11})^{\top}$ and $\alp_2 = (\alpha_{20}, \alpha_{21})^{\top}$.
    
\begin{table}
	\caption{\textit{Simulation study for Choice-I with 10\% zeros in both $\theta_{Y}$ and $\theta_{X}$}}
	\footnotesize
	\centering
	\begin{tabular}{|l| l l l l|l| l l l l|}
		\hline
		\multicolumn{10}{|c|}{$n = 50$} \\
		\hline
		Parameters &\ Mean &\ SE &\ RB & CP & Parameters &\ Mean &\ SE &\ RB & CP \\ \hline
		$\beta_{10}$ = 8.3 &\ 8.639 & 0.948 &\ 0.041 & 0.92 & $\beta_{20}$ = -1.3 & -1.310 & 0.448 & 0.008 & 0.95 \\
		$\beta_{11}$ = 4.6 &\ 4.784 & 0.611 &\ 0.040 & 0.94 & $\beta_{21}$ = 0.8 &\ 0.842 & 0.172 & 0.053 & 0.93 \\
		$\beta_{1C}$ = 6.5 &\ 6.796 & 0.828 &\ 0.046 & 0.92 & $\beta_{2C}$ = 2.1 &\ 2.167 & 0.500 & 0.032 & 0.95 \\
		$\beta_{1S}$ = -5.1 & -5.229 & 1.207 &\ 0.025 & 0.97 & $\beta_{2S}$ = 2.4 &\ 2.449 & 0.898 & 0.020 & 0.95 \\
		$\rho$ = 0.5 &\ 0.393 & 0.466 & -0.215 & 0.94 & $\sigma_{2}^2$ = 2 &\ 2.535 & 1.218 & 0.268 & 0.98 \\
		$\alpha_{10}$ = -6.4 & -6.757 & 1.232 &\ 0.056 & 0.94 & $\alpha_{20}$ = 1.8 &\ 1.871 & 0.541 & 0.040 & 0.95 \\
		$\alpha_{11}$ = 4.5 &\ 4.750 & 0.759 &\ 0.055 & 0.94 & $\alpha_{21}$ = -0.8 & -0.836 & 0.238 & 0.045 & 0.93 \\ \hline
		\hline
		\multicolumn{10}{|c|}{$n = 100$} \\
		\hline
		Parameters &\ Mean &\ SE &\ RB & CP & Parameters &\ Mean &\ SE &\ RB & CP \\ \hline
		$\beta_{10}$ = 8.3 &\ 8.457 & 0.621 &\ 0.019 & 0.94 & $\beta_{20}$ = -1.3 & -1.336 & 0.303 & 0.028 & 0.96 \\
		$\beta_{11}$ = 4.6 &\ 4.686 & 0.412 &\ 0.019 & 0.94 & $\beta_{21}$ = 0.8 &\ 0.816 & 0.119 & 0.020 & 0.93 \\
		$\beta_{1C}$ = 6.5 &\ 6.648 & 0.531 &\ 0.023 & 0.93 & $\beta_{2C}$ = 2.1 &\ 2.160 & 0.339 & 0.028 & 0.95 \\
		$\beta_{1S}$ = -5.1 & -5.140 & 0.784 &\ 0.008 & 0.94 & $\beta_{2S}$ = 2.4 &\ 2.478 & 0.614 & 0.032 & 0.94 \\
		$\rho$ = 0.5 &\ 0.467 & 0.329 & -0.066 & 0.94 & $\sigma_{2}^2$ = 2 &\ 2.272 & 0.758 & 0.136 & 0.96 \\
		$\alpha_{10}$ = -6.4 & -6.534 & 0.845 &\ 0.021 & 0.94 & $\alpha_{20}$ = 1.8 &\ 1.846 & 0.372 & 0.026 & 0.95 \\
		$\alpha_{11}$ = 4.5 &\ 4.588 & 0.513 &\ 0.019 & 0.94 & $\alpha_{21}$ = -0.8 & -0.822 & 0.165 & 0.028 & 0.95 \\ \hline
		\hline
		\multicolumn{10}{|c|}{$n = 500$} \\
		\hline
		Parameters &\ Mean &\ SE &\ RB & CP & Parameters &\ Mean &\ SE &\ RB & CP \\ \hline
		$\beta_{10}$ = 8.3 &\ 8.335 & 0.265 &\ 0.004 & 0.96 & $\beta_{20}$ = -1.3 & -1.313 & 0.133 & 0.010 & 0.94 \\
		$\beta_{11}$ = 4.6 &\ 4.620 & 0.178 &\ 0.004 & 0.95 & $\beta_{21}$ = 0.8 &\ 0.804 & 0.051 & 0.005 & 0.95 \\
		$\beta_{1C}$ = 6.5 &\ 6.530 & 0.230 &\ 0.005 & 0.94 & $\beta_{2C}$ = 2.1 &\ 2.115 & 0.147 & 0.007 & 0.94 \\
		$\beta_{1S}$ = -5.1 & -5.123 & 0.329 &\ 0.004 & 0.95 & $\beta_{2S}$ = 2.4 &\ 2.416 & 0.263 & 0.007 & 0.94 \\
		$\rho$ = 0.5 &\ 0.499 & 0.146 & -0.002 & 0.97 & $\sigma_{2}^2$ = 2 &\ 2.046 & 0.243 & 0.023 & 0.94 \\
		$\alpha_{10}$ = -6.4 & -6.413 & 0.365 &\ 0.002 & 0.94 & $\alpha_{20}$ = 1.8 &\ 1.801 & 0.167 & 0.001 & 0.94 \\
		$\alpha_{11}$ = 4.5 &\ 4.510 & 0.220 &\ 0.002 & 0.95 & $\alpha_{21}$ = -0.8 & -0.800 & 0.073 & 0.000 & 0.95 \\ \hline
	\end{tabular}
	\label{TABLE:LCRM:SimStudy-Resp10-Cov10-conjugate}
\end{table}
\begin{table}
	\caption{\textit{Simulation study for Choice-I with 35\% zeros in $\theta_{Y}$ and 30\% zeros in $\theta_{X}$}}
	\footnotesize
	\centering
	\begin{tabular}{|l| l l l l|l| l l l l|}
		\hline
		\multicolumn{10}{|c|}{$n = 50$} \\
		\hline
		Parameters &\ Mean &\ SE &\ RB & CP & Parameters &\ Mean &\ SE &\ RB & CP \\ \hline
		$\beta_{10}$ = 13.5 & 13.948 & 1.942 &\ 0.033 & 0.96 & $\beta_{20}$ = -1.3 & -1.408 & 0.572 & 0.083 & 0.96 \\ 
		$\beta_{11}$ = -8.6 & -8.975 & 1.154 &\ 0.044 & 0.94 & $\beta_{21}$ = 0.5 &\ 0.528 & 0.155 & 0.056 & 0.92 \\ 
		$\beta_{1C}$ = 9.2 &\ 9.825 & 1.769 &\ 0.068 & 0.95 & $\beta_{2C}$ = 1.2 &\ 1.309 & 0.549 & 0.091 & 0.95 \\ 
		$\beta_{1S}$ = -8.1 & -7.546 & 3.886 & -0.068 & 0.98 & $\beta_{2S}$ = 2.4 &\ 2.620 & 1.306 & 0.092 & 0.96 \\ 
		$\rho$ = 0.9 &\ 0.838 & 0.235 & -0.069 & 1.00 & $\sigma_{2}^2$ = 1 &\ 1.853 & 1.609 & 0.853 & 1.00 \\ 
		$\alpha_{10}$ = -8.4 & -8.559 & 1.617 &\ 0.019 & 0.94 & $\alpha_{20}$ = 1.8 &\ 1.853 & 0.550 & 0.030 & 0.94 \\ 
		$\alpha_{11}$ = 10.5 & 10.739 & 1.582 &\ 0.023 & 0.94 & $\alpha_{21}$ = -0.8 & -0.822 & 0.249 & 0.028 & 0.95 \\ \hline
		\hline
		\multicolumn{10}{|c|}{$n = 100$} \\
		\hline
		Parameters &\ Mean &\ SE &\ RB & CP & Parameters &\ Mean &\ SE &\ RB & CP \\ \hline
		$\beta_{10}$ = 13.5 & 13.950 & 1.380 &\ 0.033 & 0.95 & $\beta_{20}$ = -1.3 & -1.331 & 0.369 & 0.024 & 0.97 \\ 
		$\beta_{11}$ = -8.6 & -8.899 & 0.813 &\ 0.035 & 0.92 & $\beta_{21}$ = 0.5 &\ 0.526 & 0.109 & 0.051 & 0.93 \\ 
		$\beta_{1C}$ = 9.2 &\ 9.698 & 1.211 &\ 0.054 & 0.94 & $\beta_{2C}$ = 1.2 &\ 1.254 & 0.361 & 0.045 & 0.97 \\ 
		$\beta_{1S}$ = -8.1 & -7.941 & 2.657 & -0.020 & 0.96 & $\beta_{2S}$ = 2.4 &\ 2.513 & 0.865 & 0.047 & 0.96 \\ 
		$\rho$ = 0.9 &\ 0.887 & 0.141 & -0.015 & 1.00 & $\sigma_{2}^2$ = 1 &\ 1.723 & 1.390 & 0.723 & 1.00 \\ 
		$\alpha_{10}$ = -8.4 & -8.471 & 1.111 &\ 0.008 & 0.94 & $\alpha_{20}$ = 1.8 &\ 1.829 & 0.380 & 0.016 & 0.93 \\ 
		$\alpha_{11}$ = 10.5 & 10.611 & 1.084 &\ 0.011 & 0.96 & $\alpha_{21}$ = -0.8 & -0.816 & 0.171 & 0.020 & 0.95 \\ \hline
		\hline
		\multicolumn{10}{|c|}{$n = 500$} \\
		\hline
		Parameters &\ Mean &\ SE &\ RB & CP & Parameters &\ Mean &\ SE &\ RB & CP \\ \hline
		$\beta_{10}$ = 13.5 & 13.658 & 0.605 &\ 0.012 & 0.92 & $\beta_{20}$ = -1.3 & -1.302 & 0.156 &\ 0.001 & 0.95 \\
		$\beta_{11}$ = -8.6 & -8.665 & 0.355 &\ 0.008 & 0.93 & $\beta_{21}$ = 0.5 &\ 0.509 & 0.049 &\ 0.018 & 0.94 \\
		$\beta_{1C}$ = 9.2 &\ 9.297 & 0.471 &\ 0.011 & 0.96 & $\beta_{2C}$ = 1.2 &\ 1.213 & 0.152 &\ 0.011 & 0.95 \\
		$\beta_{1S}$ = -8.1 & -8.175 & 1.063 &\ 0.009 & 0.95 & $\beta_{2S}$ = 2.4 &\ 2.411 & 0.363 &\ 0.005 & 0.95 \\
		$\rho$ = 0.9 &\ 0.892 & 0.077 & -0.009 & 1.00 & $\sigma_{2}^2$ = 1 &\ 1.175 & 0.462 &\ 0.175 & 1.00 \\
		$\alpha_{10}$ = -8.4 & -8.421 & 0.493 &\ 0.003 & 0.97 & $\alpha_{20}$ = 1.8 &\ 1.799 & 0.162 & -0.001 & 0.96 \\
		$\alpha_{11}$ = 10.5 & 10.525 & 0.477 &\ 0.002 & 0.97 & $\alpha_{21}$ = -0.8 & -0.799 & 0.074 & -0.002 & 0.95 \\ \hline
	\end{tabular}
	\label{TABLE:LCRM:SimStudy-Resp35-Cov30-conjugate}
\end{table}
    
For all the studies, we generate $x_{1ij}$ and $v_{1i}$ from $\mathcal{N}(0,1)$ and $\mathcal{N}(2,1)$, respectively. We fix $\delYtwo = \delXtwo = 0.035$ radians ($2^\circ$) in all the studies and choose several sets of parameter values so that the approximate percentage of zeros in the longitudinal circular response and the circular covariate are given by (10\%, 10\%), (35\%, 30\%), (15\%, 0\%), (0\%, 15\%), and (0\%, 0\%). We consider a set of conjugate priors as discussed in Section \ref{SEC:LCRM:Bayes-estimation} with $\BMu_{\bet_1} = \BMu_{\bet_2} = \boldsymbol{0}$, $\BSig_{\bet_1} = \BSig_{\bet_2} = 100\BI_4$, $\BMu_{\alp_1} = \BMu_{\alp_2} = \boldsymbol{0}$, $\BSig_{\alp_1} = \BSig_{\alp_2} = 100\BI_2$, $\lambda_0 = 1$, $\nu_0 = 1$, and $\kappa_0 = 0.01$. The aforementioned choices of hyperparameters induce a high variability in the prior distributions which only provide vague information about the parameters. We generate 1,00,000 posterior samples from the full conditional densities of the associated model parameters using the Gibbs sampling algorithm as given in Section \ref{SEC:LCRM:Inferencial-Methodology} and discard the first 40,000 samples as a burn-in period. We also calculate the posterior mean, standard error (SE) and relative bias (RB) of all the parameters based on every 10th iterate from the remaining samples. The whole procedure is repeated 500 times, and finally, the average estimates are provided in Tables \ref{TABLE:LCRM:SimStudy-Resp10-Cov10-conjugate}-\ref{TABLE:LCRM:SimStudy-Resp0-Cov0-conjugate}. As we expect, the standard errors and biases are decreasing with increasing sample sizes. In addition, we report the coverage probability (CP) corresponding to all the parameters. We observe that the CPs for most of the parameters are around 95\%.

\begin{table}[!ht]
	\caption{\textit{Simulation study for Choice-I with 15\% zeros in $\theta_{Y}$ only}}
	\footnotesize
	\centering
	\begin{tabular}{|l| l l l l|l| l l l l|}
		\hline
		\multicolumn{10}{|c|}{$n = 50$} \\
		\hline
		Parameters &\ Mean &\ SE &\ RB & CP & Parameters &\ Mean &\ SE &\ RB & CP \\ \hline
		$\beta_{10}$ = 10.2 & 10.818 & 1.872 &\ 0.061 & 0.92 & $\beta_{20}$ = -1.3 & -1.348 & 1.894 & 0.037 & 0.95 \\
		$\beta_{11}$ = -5.6 & -5.794 & 0.675 &\ 0.035 & 0.93 & $\beta_{21}$ = 0.8 &\ 0.829 & 0.181 & 0.037 & 0.94 \\
		$\beta_{1C}$ = 2.5 &\ 2.324 & 1.041 & -0.071 & 0.93 & $\beta_{2C}$ = 2.6 &\ 2.681 & 1.330 & 0.031 & 0.94 \\
		$\beta_{1S}$ = 2.1 &\ 2.176 & 2.568 &\ 0.036 & 0.94 & $\beta_{2S}$ = 2.4 &\ 2.488 & 2.924 & 0.037 & 0.96 \\
		$\rho$ = 0 & -0.012 & 0.386 & -0.012$^*$ & 1.00 & $\sigma_{2}^2$ = 8 &\ 9.001 & 3.157 & 0.125 & 0.95 \\
		$\alpha_{10}$ = -8.4 & -8.585 & 1.491 &\ 0.022 & 0.95 & $\alpha_{20}$ = 1.8 &\ 1.900 & 0.590 & 0.056 & 0.93 \\
		$\alpha_{11}$ = 10.5 & 10.788 & 1.566 &\ 0.027 & 0.93 & $\alpha_{21}$ = 1.5 &\ 1.528 & 0.334 & 0.018 & 0.95 \\ \hline
		\hline
		\multicolumn{10}{|c|}{$n = 100$} \\
		\hline
		Parameters &\ Mean &\ SE &\ RB & CP & Parameters &\ Mean &\ SE &\ RB & CP \\ \hline
		$\beta_{10}$ = 10.2 & 10.380 & 1.091 &\ 0.018 & 0.95 & $\beta_{20}$ = -1.3 & -1.278 & 1.227 & -0.017 & 0.95 \\
		$\beta_{11}$ = -5.6 & -5.682 & 0.447 &\ 0.015 & 0.94 & $\beta_{21}$ = 0.8 &\ 0.816 & 0.120 &\ 0.020 & 0.97 \\
		$\beta_{1C}$ = 2.5 &\ 2.464 & 0.568 & -0.014 & 0.95 & $\beta_{2C}$ = 2.6 &\ 2.591 & 0.841 & -0.004 & 0.95 \\
		$\beta_{1S}$ = 2.1 &\ 2.284 & 1.670 &\ 0.087 & 0.93 & $\beta_{2S}$ = 2.4 &\ 2.422 & 1.985 &\ 0.009 & 0.96 \\
		$\rho$ = 0 &\ 0.001 & 0.378 &\ 0.001$^*$ & 1.00 & $\sigma_{2}^2$ = 8 &\ 8.464 & 2.085 &\ 0.058 & 0.95 \\
		$\alpha_{10}$ = -8.4 & -8.570 & 1.016 &\ 0.020 & 0.94 & $\alpha_{20}$ = 1.8 &\ 1.847 & 0.398 &\ 0.026 & 0.94 \\
		$\alpha_{11}$ = 10.5 & 10.737 & 1.072 &\ 0.023 & 0.94 & $\alpha_{21}$ = 1.5 &\ 1.536 & 0.228 &\ 0.024 & 0.94 \\ \hline
		\hline
		\multicolumn{10}{|c|}{$n = 500$} \\
		\hline
		Parameters &\ Mean &\ SE &\ RB & CP & Parameters &\ Mean &\ SE &\ RB & CP \\ \hline
		$\beta_{10}$ = 10.2 & 10.215 & 0.461 &\ 0.001 & 0.93 & $\beta_{20}$ = -1.3 & -1.296 & 0.515 & -0.003 & 0.96 \\
		$\beta_{11}$ = -5.6 & -5.611 & 0.203 &\ 0.002 & 0.94 & $\beta_{21}$ = 0.8 &\ 0.804 & 0.055 &\ 0.005 & 0.94 \\
		$\beta_{1C}$ = 2.5 &\ 2.496 & 0.222 & -0.001 & 0.95 & $\beta_{2C}$ = 2.6 &\ 2.610 & 0.353 &\ 0.004 & 0.95 \\
		$\beta_{1S}$ = 2.1 &\ 2.165 & 0.683 &\ 0.031 & 0.94 & $\beta_{2S}$ = 2.4 &\ 2.375 & 0.874 & -0.010 & 0.96 \\
		$\rho$ = 0 & -0.003 & 0.293 & -0.003$^*$ & 0.97 & $\sigma_{2}^2$ = 8 &\ 8.121 & 0.891 &\ 0.015 & 0.93 \\
		$\alpha_{10}$ = -8.4 & -8.429 & 0.447 &\ 0.003 & 0.95 & $\alpha_{20}$ = 1.8 &\ 1.813 & 0.176 &\ 0.007 & 0.93 \\
		$\alpha_{11}$ = 10.5 & 10.528 & 0.470 &\ 0.003 & 0.95 & $\alpha_{21}$ = 1.5 &\ 1.501 & 0.100 &\ 0.001 & 0.93 \\ \hline
	\end{tabular}
	\label{TABLE:LCRM:SimStudy-Resp15-Cov0-conjugate}
	\begin{tablenotes}
		\item \: \: \: \: $^*$Bias is given instead of RB.  
	\end{tablenotes}
\end{table}
	
\begin{table}[ht!]
	\caption{\textit{Simulation study for Choice-I with 15\% zeros in $\theta_{X}$ only}}
	\footnotesize
	\centering
	\begin{tabular}{|l| l l l l|l| l l l l|}
		\hline
		\multicolumn{10}{|c|}{$n = 50$} \\
		\hline
		Parameters &\ Mean &\ SE &\ RB & CP & Parameters &\ Mean &\ SE &\ RB & CP \\ \hline
		$\beta_{10}$ = 10.2 & 10.589 & 1.133 &\ 0.038 & 0.93 & $\beta_{20}$ = 5.2 &\ 5.341 & 0.995 & 0.027 & 0.95 \\
		$\beta_{11}$ = -5.6 & -5.764 & 0.607 &\ 0.029 & 0.94 & $\beta_{21}$ = 3.8 &\ 3.942 & 0.426 & 0.037 & 0.93 \\
		$\beta_{1C}$ = 2.5 &\ 2.500 & 0.626 &\ 0.000 & 0.95 & $\beta_{2C}$ = 2.6 &\ 2.687 & 0.967 & 0.034 & 0.95 \\
		$\beta_{1S}$ = 2.1 &\ 2.087 & 1.639 & -0.006 & 0.95 & $\beta_{2S}$ = 2.4 &\ 2.429 & 2.500 & 0.012 & 0.94 \\
		$\rho$ = 0 & -0.068 & 0.448 & -0.068$^*$ & 0.96 & $\sigma_{2}^2$ = 5 &\ 5.437 & 2.098 & 0.087 & 0.96 \\
		$\alpha_{10}$ = -8.4 & -8.635 & 1.690 &\ 0.028 & 0.95 & $\alpha_{20}$ = 1.5 &\ 1.520 & 0.519 & 0.013 & 0.94 \\
		$\alpha_{11}$ = 10.5 & 10.846 & 1.627 &\ 0.033 & 0.95 & $\alpha_{21}$ = -1.2 & -1.228 & 0.272 & 0.023 & 0.95 \\ \hline
		\hline
		\multicolumn{10}{|c|}{$n = 100$} \\
		\hline
		Parameters &\ Mean &\ SE &\ RB & CP & Parameters &\ Mean &\ SE &\ RB & CP \\ \hline
		$\beta_{10}$ = 10.2 & 10.414 & 0.772 &\ 0.021 & 0.92 & $\beta_{20}$ = 5.2 &\ 5.268 & 0.673 & 0.013 & 0.94 \\
		$\beta_{11}$ = -5.6 & -5.709 & 0.432 &\ 0.019 & 0.92 & $\beta_{21}$ = 3.8 &\ 3.885 & 0.298 & 0.022 & 0.94 \\
		$\beta_{1C}$ = 2.5 &\ 2.531 & 0.403 &\ 0.012 & 0.95 & $\beta_{2C}$ = 2.6 &\ 2.702 & 0.649 & 0.039 & 0.96 \\
		$\beta_{1S}$ = 2.1 &\ 2.107 & 1.058 &\ 0.004 & 0.95 & $\beta_{2S}$ = 2.4 &\ 2.476 & 1.650 & 0.032 & 0.97 \\
		$\rho$ = 0 & -0.043 & 0.386 & -0.043$^*$ & 0.95 & $\sigma_{2}^2$ = 5 &\ 5.262 & 1.454 & 0.052 & 0.96 \\
		$\alpha_{10}$ = -8.4 & -8.480 & 1.193 &\ 0.009 & 0.93 & $\alpha_{20}$ = 1.5 &\ 1.524 & 0.374 & 0.016 & 0.91 \\
		$\alpha_{11}$ = 10.5 & 10.618 & 1.126 &\ 0.011 & 0.94 & $\alpha_{21}$ = -1.2 & -1.215 & 0.193 & 0.012 & 0.93 \\ \hline
		\hline
		\multicolumn{10}{|c|}{$n = 500$} \\
		\hline
		Parameters &\ Mean &\ SE &\ RB & CP & Parameters &\ Mean &\ SE &\ RB & CP \\ \hline
		$\beta_{10}$ = 10.2 & 10.241 & 0.323 &\ 0.004 & 0.93 & $\beta_{20}$ = 5.2 &\ 5.218 & 0.294 & 0.003 & 0.94 \\
		$\beta_{11}$ = -5.6 & -5.615 & 0.182 &\ 0.003 & 0.95 & $\beta_{21}$ = 3.8 &\ 3.818 & 0.129 & 0.005 & 0.94 \\
		$\beta_{1C}$ = 2.5 &\ 2.503 & 0.165 &\ 0.001 & 0.95 & $\beta_{2C}$ = 2.6 &\ 2.607 & 0.281 & 0.003 & 0.94 \\
		$\beta_{1S}$ = 2.1 &\ 2.098 & 0.423 & -0.001 & 0.95 & $\beta_{2S}$ = 2.4 &\ 2.405 & 0.705 & 0.002 & 0.95 \\
		$\rho$ = 0 & -0.013 & 0.208 & -0.013$^*$ & 0.95 & $\sigma_{2}^2$ = 5 &\ 5.091 & 0.622 & 0.018 & 0.95 \\
		$\alpha_{10}$ = -8.4 & -8.454 & 0.505 &\ 0.006 & 0.96 & $\alpha_{20}$ = 1.5 &\ 1.504 & 0.156 & 0.003 & 0.96 \\
		$\alpha_{11}$ = 10.5 & 10.552 & 0.490 &\ 0.005 & 0.96 & $\alpha_{21}$ = -1.2 & -1.204 & 0.081 & 0.003 & 0.95 \\ \hline
	\end{tabular}
	\label{TABLE:LCRM:SimStudy-Resp0-Cov15-conjugate}
	\begin{tablenotes}
		\item \: \: \: \: $^*$Bias is given instead of RB.  
	\end{tablenotes}
\end{table}
\begin{table}
	\caption{\textit{Simulation study for Choice-I under without any zero-inflation}}
	\footnotesize
	\centering
	\begin{tabular}{|l| l l l l|l| l l l l|}
		\hline
		\multicolumn{10}{|c|}{$n = 50$} \\
		\hline
		Parameters &\ Mean &\ SE &\ RB & CP & Parameters & Mean &\ SE &\ RB & CP \\ \hline
		$\beta_{10}$ = 5.3 &\ 5.558 & 2.477 & 0.049 & 0.98 & $\beta_{20}$ = 2.5 & 2.652 & 1.498 & 0.061 & 0.95 \\
		$\beta_{11}$ = 4.6 &\ 4.798 & 0.565 & 0.043 & 0.92 & $\beta_{21}$ = 0.8 & 0.817 & 0.260 & 0.022 & 0.94 \\
		$\beta_{1C}$ = 2.5 &\ 2.589 & 1.144 & 0.035 & 0.98 & $\beta_{2C}$ = 2.6 & 2.720 & 0.711 & 0.046 & 0.95 \\
		$\beta_{1S}$ = 2.1 &\ 2.256 & 2.953 & 0.074 & 0.98 & $\beta_{2S}$ = 2.4 & 2.507 & 1.831 & 0.045 & 0.95 \\
		$\rho$ = 0.8 &\ 0.808 & 0.203 & 0.010 & 0.95 & $\sigma_{2}^2$ = 1 & 1.688 & 1.392 & 0.688 & 1.00 \\
		$\alpha_{10}$ = -5.4 & -5.628 & 0.957 & 0.042 & 0.93 & $\alpha_{20}$ = 1.8 & 1.915 & 0.673 & 0.064 & 0.95 \\
		$\alpha_{11}$ = 3.5 &\ 3.652 & 0.565 & 0.044 & 0.93 & $\alpha_{21}$ = 1.5 & 1.557 & 0.385 & 0.038 & 0.95 \\ \hline
		\hline
		\multicolumn{10}{|c|}{$n = 100$} \\
		\hline
		Parameters &\ Mean &\ SE &\ RB & CP & Parameters & Mean &\ SE &\ RB & CP \\ \hline
		$\beta_{10}$ = 5.3 &\ 5.415 & 1.657 & 0.022 & 0.96 & $\beta_{20}$ = 2.5 & 2.573 & 0.989 &\ 0.029 & 0.95 \\
		$\beta_{11}$ = 4.6 &\ 4.695 & 0.372 & 0.021 & 0.92 & $\beta_{21}$ = 0.8 & 0.799 & 0.176 & -0.002 & 0.94 \\
		$\beta_{1C}$ = 2.5 &\ 2.543 & 0.743 & 0.017 & 0.97 & $\beta_{2C}$ = 2.6 & 2.657 & 0.478 &\ 0.022 & 0.96 \\
		$\beta_{1S}$ = 2.1 &\ 2.203 & 1.999 & 0.049 & 0.95 & $\beta_{2S}$ = 2.4 & 2.466 & 1.202 &\ 0.027 & 0.96 \\
		$\rho$ = 0.8 &\ 0.814 & 0.151 & 0.018 & 0.96 & $\sigma_{2}^2$ = 1 & 1.434 & 0.929 &\ 0.434 & 0.99 \\
		$\alpha_{10}$ = -5.4 & -5.557 & 0.663 & 0.029 & 0.94 & $\alpha_{20}$ = 1.8 & 1.844 & 0.466 &\ 0.025 & 0.94 \\
		$\alpha_{11}$ = 3.5 &\ 3.597 & 0.390 & 0.028 & 0.93 & $\alpha_{21}$ = 1.5 & 1.536 & 0.269 &\ 0.024 & 0.94 \\ \hline
		\hline
		\multicolumn{10}{|c|}{$n = 500$} \\
		\hline
		Parameters &\ Mean &\ SE &\ RB & CP & Parameters & Mean &\ SE &\ RB & CP \\ \hline
		$\beta_{10}$ = 5.3 &\ 5.326 & 0.670 & 0.005 & 0.94 & $\beta_{20}$ = 2.5 & 2.525 & 0.415 &\ 0.010 & 0.96 \\
		$\beta_{11}$ = 4.6 &\ 4.618 & 0.155 & 0.004 & 0.96 & $\beta_{21}$ = 0.8 & 0.799 & 0.073 & -0.002 & 0.96 \\
		$\beta_{1C}$ = 2.5 &\ 2.524 & 0.299 & 0.010 & 0.96 & $\beta_{2C}$ = 2.6 & 2.617 & 0.200 &\ 0.006 & 0.96 \\
		$\beta_{1S}$ = 2.1 &\ 2.110 & 0.812 & 0.005 & 0.95 & $\beta_{2S}$ = 2.4 & 2.396 & 0.509 & -0.002 & 0.95 \\
		$\rho$ = 0.8 &\ 0.802 & 0.073 & 0.002 & 0.98 & $\sigma_{2}^2$ = 1 & 1.047 & 0.178 &\ 0.047 & 0.97 \\
		$\alpha_{10}$ = -5.4 & -5.426 & 0.282 & 0.005 & 0.96 & $\alpha_{20}$ = 1.8 & 1.815 & 0.202 &\ 0.009 & 0.95 \\
		$\alpha_{11}$ = 3.5 &\ 3.517 & 0.166 & 0.005 & 0.96 & $\alpha_{21}$ = 1.5 & 1.506 & 0.116 &\ 0.004 & 0.95 \\ \hline
	\end{tabular}
	\label{TABLE:LCRM:SimStudy-Resp0-Cov0-conjugate}
\end{table}

\subsection{Model comparison} \label{SEC:LCRM:Model-Comparison}
Three distinct models are taken into consideration for comparison: Model-I, Model-II and Model-III. Model-I, given by (\ref{EQN:LCRM:ZeroInflated-Stage-I}) and (\ref{EQN:LCRM:ZeroInflated-Stage-II}), considers zero-inflation in both the longitudinal circular response and the circular covariate, whereas Model-II does not consider zero-inflation in both of them, and Model-III considers zero-inflation only in the longitudinal circular response, but not in the circular covariate. At first, we evaluate how well Model-I, II and III are estimating the parameters when the data is generated from Model-I. For this case, we consider $n = 200$ and $\delYtwo = \delXtwo = 0.14$ radians ($8^\circ$). The conjugate priors are used with the same hyperparameters as discussed before. In Table \ref{TABLE:LCRM:Model-Comparison-Conjugate}, we provide the averages of the parameter estimates over 500 replications. One can observe that RBs of all the parameters are lower for both the stages in Model-I with around 95\% CPs. However, for most of the parameters, the biases are higher and the CPs are lower in Model-II as well as in Model-III compared to Model-I. Note that the estimates associated with the regression of $\tXi$ on $v_{1i}$ based on Model-II and III, where zero-inflation is not considered in the circular covariate, are comparatively more biased than Model-I which incorporates zero-inflation in the circular covariate. In the presence of zero-inflation in the longitudinal circular response and the circular covariate, the results from simulation study in Table \ref{TABLE:LCRM:Model-Comparison-Conjugate} indicate that both Model-II and Model-III might not be suitable to represent the underlying relationship in the data.
	
\begin{table}
	\caption{\textit{Simulation study with 45\% zeros in $\theta_{Y}$ and 55\% zeros in $\theta_{X}$}}
	\footnotesize
	\centering
	\begin{tabular}{|l| l l l l|l| l l l l|}
		\hline
		\multicolumn{10}{|c|}{Model-I} \\
		\hline
		Parameters &\ Mean &\ SE &\ RB & CP & Parameters &\ Mean &\ SE &\ RB & CP \\ \hline
		$\beta_{10}$ = 8.8 &\ 8.760 & 3.489 & -0.005 & 0.96 & $\beta_{20}$ = -1.6 & -1.606 & 3.662 &\ 0.004 & 0.97 \\
		$\beta_{11}$ = 5.2 &\ 5.269 & 0.348 &\ 0.013 & 0.95 & $\beta_{21}$ = 0.8 &\ 0.813 & 0.106 &\ 0.017 & 0.94 \\
		$\beta_{1C}$ = 1.5 &\ 1.661 & 3.564 &\ 0.107 & 0.95 & $\beta_{2C}$ = 1.2 &\ 1.175 & 3.777 & -0.021 & 0.97 \\
		$\beta_{1S}$ = 1.2 &\ 1.046 & 1.521 & -0.129 & 0.95 & $\beta_{2S}$ = 1.8 &\ 1.765 & 1.630 & -0.020 & 0.95 \\
		$\rho$ = 0.8 &\ 0.661 & 0.228 & -0.174 & 0.99 & $\sigma_{2}^2$ = 5 &\ 5.249 & 1.070 &\ 0.050 & 0.97 \\
		$\alpha_{10}$ = 3.4 &\ 3.465 & 0.649 &\ 0.019 & 0.95 & $\alpha_{20}$ = -1.2 & -1.203 & 0.256 &\ 0.002 & 0.95 \\
		$\alpha_{11}$ = 4.5 &\ 4.567 & 0.606 &\ 0.015 & 0.95 & $\alpha_{21}$ = 1.3 &\ 1.314 & 0.164 &\ 0.011 & 0.95 \\ \hline
		\hline
		\multicolumn{10}{|c|}{Model-II} \\
		\hline
		Parameters &\ Mean &\ SE &\ RB & CP & Parameters &\ Mean &\ SE &\ RB & CP \\ \hline
		$\beta_{10}$ = 8.8 &\ 8.165 & 3.168 & -0.072 & 0.96 & $\beta_{20}$ = -1.6 & -1.779 & 3.162 &\ 0.112 & 0.94 \\
		$\beta_{11}$ = 5.2 &\ 4.770 & 0.490 & -0.083 & 0.43 & $\beta_{21}$ = 0.8 &\ 0.638 & 0.175 & -0.202 & 0.25 \\
		$\beta_{1C}$ = 1.5 &\ 1.331 & 3.160 & -0.113 & 0.96 & $\beta_{2C}$ = 1.2 &\ 1.445 & 3.240 &\ 0.205 & 0.94 \\
		$\beta_{1S}$ = 1.2 &\ 0.972 & 1.473 & -0.190 & 0.96 & $\beta_{2S}$ = 1.8 &\ 1.335 & 1.469 & -0.258 & 0.90 \\
		$\rho$ = 0.8 &\ 0.692 & 0.201 & -0.135 & 0.98 & $\sigma_{2}^2$ = 5 &\ 3.837 & 1.430 & -0.233 & 0.56 \\
		$\alpha_{10}$ = 3.4 &\ 3.810 & 0.797 &\ 0.120 & 0.84 & $\alpha_{20}$ = -1.2 & -0.938 & 0.334 & -0.218 & 0.67 \\
		$\alpha_{11}$ = 4.5 &\ 2.871 & 1.655 & -0.362 & 0.00 & $\alpha_{21}$ = 1.3 &\ 0.833 & 0.476 & -0.360 & 0.00 \\ \hline
		\hline
		\multicolumn{10}{|c|}{Model-III} \\
		\hline
		Parameters &\ Mean &\ SE &\ RB & CP & Parameters &\ Mean &\ SE &\ RB & CP \\ \hline
		$\beta_{10}$ = 8.8 &\ 8.639 & 3.293 & -0.018 & 0.96 & $\beta_{20}$ = -1.6 & -1.942 & 3.498 &\ 0.214 & 0.97 \\
		$\beta_{11}$ = 5.2 &\ 5.278 & 0.351 &\ 0.015 & 0.94 & $\beta_{21}$ = 0.8 &\ 0.815 & 0.102 &\ 0.019 & 0.95 \\
		$\beta_{1C}$ = 1.5 &\ 1.834 & 3.346 &\ 0.223 & 0.95 & $\beta_{2C}$ = 1.2 &\ 1.574 & 3.581 &\ 0.312 & 0.97 \\
		$\beta_{1S}$ = 1.2 &\ 0.945 & 1.558 & -0.212 & 0.96 & $\beta_{2S}$ = 1.8 &\ 1.675 & 1.611 & -0.070 & 0.95 \\
		$\rho$ = 0.8 &\ 0.681 & 0.217 & -0.149 & 0.99 & $\sigma_{2}^2$ = 5 &\ 5.301 & 1.143 &\ 0.060 & 0.96 \\
		$\alpha_{10}$ = 3.4 &\ 3.877 & 0.842 &\ 0.140 & 0.81 & $\alpha_{20}$ = -1.2 & -0.955 & 0.330 & -0.204 & 0.67 \\
		$\alpha_{11}$ = 4.5 &\ 2.836 & 1.690 & -0.370 & 0.00 & $\alpha_{21}$ = 1.3 &\ 0.839 & 0.469 & -0.354 & 0.00 \\ \hline
	\end{tabular}
	\label{TABLE:LCRM:Model-Comparison-Conjugate}
\end{table}

\subsection{Sensitivity analysis} \label{SEC:LCRM:Sensitivity-Analysis}
\subsubsection{Sensitivity to model misspecification} \label{SEC:LCRM:Model-Misspecification}
We consider a misspecified model to perform the sensitivity analysis of our proposed method. In this context, at first, $\tYij$ and $\tXi$ are generated and censored as previously discussed. After that, they are contaminated with random zeros with probability $\eta_Y$ and $\eta_X$. Here, we consider $n = 200$, $\delYtwo = \delXtwo = 0.14$ radians ($8^\circ$) and the same choice of hyperparameters as discussed in Section \ref{SEC:LCRM:Model-Comparison}. It is to be noted that, when $\eta_{Y} = \eta_{X} = 0$, this contaminated model reduces to the proposed zero-inflated LCRM model, given by (\ref{EQN:LCRM:ZeroInflated-Stage-I}) and (\ref{EQN:LCRM:ZeroInflated-Stage-II}). Subsequently, we provide the average estimates of the posterior mean, RB and SE based on 500 replications and compare the three models, Model-I, II and III, in terms of RB and CP. We report the simulation results under misspecification in Table \ref{TABLE:LCRM:Misspecified-Model-etaY=etaX=0.05} and \ref{TABLE:LCRM:Misspecified-Model-etaY=etaX=0.075} for $\eta_{Y} = \eta_{X} = 0.05$ and $\eta_{Y} = \eta_{X} = 0.075$, respectively. The results demonstrate that the performance of the estimates declines with respect to RB for all the models with increasing the proportion of contamination. In comparison to Model-II and III, the biases for majority of the parameters are significantly lower in Model-I. Specifically, for most of the parameters, the CPs are substantially higher (around 93\%) in Model-I than the other two models. In particular, according to the simulation results under Model-II and III, the CPs for most of the parameters related to the Stage-II regression are extremely low, often close to zero. For example, the CPs for $\alpha_{11}$ and $\alpha_{21}$ are approximately zero under Model-II and Model-III in Table \ref{TABLE:LCRM:Misspecified-Model-etaY=etaX=0.05} and \ref{TABLE:LCRM:Misspecified-Model-etaY=etaX=0.075}. The results indicate that the proposed model performs much better than Model-II and III, even if the data is generated from the misspecified model.

\begin{table}
	\caption{\textit{Simulation study under misspecified model with $\eta_{Y} = \eta_{X} = 0.05$}}
	\footnotesize
	\centering
	\begin{tabular}{|l| l l l l|l| l l l l|}
		\hline
		\multicolumn{10}{|c|}{Model-I} \\
		\hline
		Parameters &\ Mean &\ SE &\ RB & CP & Parameters &\ Mean &\ SE &\ RB & CP \\ \hline
		$\beta_{10}$ = 8.3 &\ 7.894 & 3.650 & -0.049 & 0.97 & $\beta_{20}$ = -1.3 & -1.463 & 2.419 &\ 0.125 & 0.94 \\
		$\beta_{11}$ = 5.2 &\ 5.038 & 0.398 & -0.031 & 0.87 & $\beta_{21}$ = 0.8 &\ 0.757 & 0.102 & -0.053 & 0.90 \\
		$\beta_{1C}$ = 1.5 &\ 1.837 & 3.741 &\ 0.225 & 0.98 & $\beta_{2C}$ = 2.1 &\ 2.248 & 2.486 &\ 0.070 & 0.94 \\
		$\beta_{1S}$ = 1.2 &\ 0.984 & 1.823 & -0.180 & 0.95 & $\beta_{2S}$ = 2.4 &\ 2.115 & 0.918 & -0.119 & 0.93 \\
		$\rho$ = 0.25 &\ 0.244 & 0.254 & -0.024 & 0.93 & $\sigma_{2}^2$ = 1 &\ 0.943 & 0.252 & -0.057 & 0.93 \\
		$\alpha_{10}$ = 3.4 &\ 3.635 & 0.707 &\ 0.069 & 0.92 & $\alpha_{20}$ = -1.2 & -1.166 & 0.267 & -0.028 & 0.93 \\
		$\alpha_{11}$ = 4.5 &\ 4.444 & 0.623 & -0.012 & 0.92 & $\alpha_{21}$ = 1.3 &\ 1.269 & 0.170 & -0.023 & 0.93 \\ \hline
		\hline
		\multicolumn{10}{|c|}{Model-II} \\
		\hline
		Parameters &\ Mean &\ SE &\ RB & CP & Parameters &\ Mean &\ SE &\ RB & CP \\ \hline
		$\beta_{10}$ = 8.3 &\ 6.377 & 3.785 & -0.232 & 0.89 & $\beta_{20}$ = -1.3 & -1.335 & 1.895 &\ 0.027 & 0.94 \\
		$\beta_{11}$ = 5.2 &\ 4.341 & 0.889 & -0.165 & 0.04 & $\beta_{21}$ = 0.8 &\ 0.454 & 0.350 & -0.432 & 0.00 \\
		$\beta_{1C}$ = 1.5 &\ 2.439 & 3.505 &\ 0.626 & 0.98 & $\beta_{2C}$ = 2.1 &\ 1.827 & 1.962 & -0.130 & 0.91 \\
		$\beta_{1S}$ = 1.2 &\ 0.515 & 1.875 & -0.571 & 0.94 & $\beta_{2S}$ = 2.4 &\ 1.603 & 1.023 & -0.332 & 0.65 \\
		$\rho$ = 0.25 & -0.225 & 0.517 & -1.899 & 0.29 & $\sigma_{2}^2$ = 1 &\ 0.619 & 0.409 & -0.381 & 0.20 \\
		$\alpha_{10}$ = 3.4 &\ 4.005 & 0.916 &\ 0.178 & 0.77 & $\alpha_{20}$ = -1.2 & -0.895 & 0.365 & -0.254 & 0.57 \\
		$\alpha_{11}$ = 4.5 &\ 2.760 & 1.764 & -0.387 & 0.00 & $\alpha_{21}$ = 1.3 &\ 0.791 & 0.516 & -0.392 & 0.00 \\ \hline
		\hline
		\multicolumn{10}{|c|}{Model-III} \\
		\hline
		Parameters &\ Mean &\ SE &\ RB & CP & Parameters &\ Mean &\ SE &\ RB & CP \\ \hline
		$\beta_{10}$ = 8.3 &\ 7.764 & 3.626 & -0.065 & 0.96 & $\beta_{20}$ = -1.3 & -1.774 & 2.354 &\ 0.365 & 0.95 \\
		$\beta_{11}$ = 5.2 &\ 5.017 & 0.398 & -0.035 & 0.87 & $\beta_{21}$ = 0.8 &\ 0.755 & 0.106 & -0.056 & 0.88 \\
		$\beta_{1C}$ = 1.5 &\ 1.960 & 3.692 &\ 0.307 & 0.96 & $\beta_{2C}$ = 2.1 &\ 2.618 & 2.412 &\ 0.247 & 0.95 \\
		$\beta_{1S}$ = 1.2 &\ 0.937 & 1.826 & -0.219 & 0.94 & $\beta_{2S}$ = 2.4 &\ 2.034 & 0.906 & -0.153 & 0.91 \\
		$\rho$ = 0.25 &\ 0.246 & 0.267 & -0.016 & 0.91 & $\sigma_{2}^2$ = 1 &\ 0.947 & 0.259 & -0.053 & 0.92 \\
		$\alpha_{10}$ = 3.4 &\ 4.011 & 0.954 &\ 0.180 & 0.74 & $\alpha_{20}$ = -1.2 & -0.882 & 0.377 & -0.265 & 0.53 \\
		$\alpha_{11}$ = 4.5 &\ 2.755 & 1.770 & -0.388 & 0.00 & $\alpha_{21}$ = 1.3 &\ 0.782 & 0.525 & -0.399 & 0.00 \\ \hline
	\end{tabular}
	\label{TABLE:LCRM:Misspecified-Model-etaY=etaX=0.05}
\end{table}
    
\begin{table}
	\caption{\textit{Simulation study under misspecified model with $\eta_{Y} = \eta_{X} = 0.075$}}
	\footnotesize
	\centering
	\begin{tabular}{|l| l l l l|l| l l l l|}
		\hline
		\multicolumn{10}{|c|}{Model-I} \\
		\hline
		Parameters &\ Mean &\ SE &\ RB & CP & Parameters &\ Mean &\ SE &\ RB & CP \\ \hline
		$\beta_{10}$ = 8.3 &\ 7.716 & 3.558 & -0.070 & 0.97 & $\beta_{20}$ = -1.3 & -1.323 & 2.283 &\ 0.018 & 0.95 \\
		$\beta_{11}$ = 5.2 &\ 4.937 & 0.444 & -0.051 & 0.79 & $\beta_{21}$ = 0.8 &\ 0.734 & 0.114 & -0.083 & 0.83 \\
		$\beta_{1C}$ = 1.5 &\ 1.936 & 3.631 &\ 0.290 & 0.98 & $\beta_{2C}$ = 2.1 &\ 2.080 & 2.350 & -0.010 & 0.95 \\
		$\beta_{1S}$ = 1.2 &\ 0.955 & 1.848 & -0.204 & 0.94 & $\beta_{2S}$ = 2.4 &\ 2.046 & 0.899 & -0.148 & 0.93 \\
		$\rho$ = 0.25 &\ 0.246 & 0.250 & -0.015 & 0.93 & $\sigma_{2}^2$ = 1 &\ 0.887 & 0.257 & -0.113 & 0.86 \\
		$\alpha_{10}$ = 3.4 &\ 3.659 & 0.707 &\ 0.076 & 0.91 & $\alpha_{20}$ = -1.2 & -1.167 & 0.269 & -0.027 & 0.94 \\
		$\alpha_{11}$ = 4.5 &\ 4.423 & 0.626 & -0.017 & 0.93 & $\alpha_{21}$ = 1.3 &\ 1.260 & 0.174 & -0.031 & 0.92 \\ \hline
		\hline
		\multicolumn{10}{|c|}{Model-II} \\
		\hline
		Parameters &\ Mean &\ SE &\ RB & CP & Parameters &\ Mean &\ SE &\ RB & CP  \\ \hline
		$\beta_{10}$ = 8.3 &\ 6.347 & 3.967 & -0.235 & 0.92 & $\beta_{20}$ = -1.3 & -1.285 & 1.970 & -0.011 & 0.93  \\
		$\beta_{11}$ = 5.2 &\ 4.274 & 0.954 & -0.178 & 0.01 & $\beta_{21}$ = 0.8 &\ 0.434 & 0.370 & -0.458 & 0.00  \\
		$\beta_{1C}$ = 1.5 &\ 2.471 & 3.672 &\ 0.647 & 0.97 & $\beta_{2C}$ = 2.1 &\ 1.763 & 2.056 & -0.160 & 0.87  \\
		$\beta_{1S}$ = 1.2 &\ 0.515 & 1.968 & -0.571 & 0.94 & $\beta_{2S}$ = 2.4 &\ 1.525 & 1.075 & -0.365 & 0.57  \\
		$\rho$ = 0.25 & -0.256 & 0.544 & -2.025 & 0.21 & $\sigma_{2}^2$ = 1 &\ 0.560 & 0.460 & -0.440 & 0.10  \\
		$\alpha_{10}$ = 3.4 &\ 3.999 & 0.914 &\ 0.176 & 0.78 & $\alpha_{20}$ = -1.2 & -0.850 & 0.401 & -0.291 & 0.45  \\
		$\alpha_{11}$ = 4.5 &\ 2.736 & 1.789 & -0.392 & 0.00 & $\alpha_{21}$ = 1.3 &\ 0.758 & 0.548 & -0.417 & 0.00 \\ \hline
		\hline
		\multicolumn{10}{|c|}{Model-III} \\
		\hline
		Parameters &\ Mean &\ SE &\ RB & CP & Parameters &\ Mean &\ SE &\ RB & CP \\ \hline
		$\beta_{10}$ = 8.3 &\ 7.865 & 3.679 & -0.052 & 0.97 & $\beta_{20}$ = -1.3 & -1.586 & 2.391 &\ 0.220 & 0.92 \\
		$\beta_{11}$ = 5.2 &\ 4.920 & 0.455 & -0.054 & 0.77 & $\beta_{21}$ = 0.8 &\ 0.736 & 0.115 & -0.080 & 0.79 \\
		$\beta_{1C}$ = 1.5 &\ 1.784 & 3.735 &\ 0.189 & 0.98 & $\beta_{2C}$ = 2.1 &\ 2.403 & 2.445 &\ 0.144 & 0.92 \\
		$\beta_{1S}$ = 1.2 &\ 1.096 & 1.838 & -0.087 & 0.95 & $\beta_{2S}$ = 2.4 &\ 1.993 & 0.934 & -0.170 & 0.88 \\
		$\rho$ = 0.25 &\ 0.225 & 0.258 & -0.100 & 0.93 & $\sigma_{2}^2$ = 1 &\ 0.872 & 0.267 & -0.128 & 0.84 \\
		$\alpha_{10}$ = 3.4 &\ 4.066 & 0.963 &\ 0.196 & 0.75 & $\alpha_{20}$ = -1.2 & -0.852 & 0.401 & -0.290 & 0.46 \\
		$\alpha_{11}$ = 4.5 &\ 2.698 & 1.827 & -0.401 & 0.00 & $\alpha_{21}$ = 1.3 &\ 0.754 & 0.552 & -0.420 & 0.00 \\ \hline
	\end{tabular}
	\label{TABLE:LCRM:Misspecified-Model-etaY=etaX=0.075}
\end{table}

\subsubsection{Sensitivity to prior specifications}\label{SEC:LCRM:Prior-Sensitivity}
We conduct a sensitivity analysis with respect to two different prior choices. We denote the two choices by Choice-I and Choice-II, where Choice-I accounts for the same hyperparameters as discussed in Section \ref{SEC:LCRM:Simulations}, and Choice-II accounts for $\BMu_{\bet_1} = \BMu_{\bet_2} = \boldsymbol{0}$, $\BSig_{\bet_1} = \BSig_{\bet_2} = 1000\BI_4$, $\BMu_{\alp_1} = \BMu_{\alp_2} = \boldsymbol{0}$, $\BSig_{\alp_1} = \BSig_{\alp_2} = 1000\BI_2$, $\lambda_0 = 1$, $\nu_0 = 0.01$, and $\kappa_0 = 0.01$. The average estimates of the parameters for Choice-I are already given in Tables \ref{TABLE:LCRM:SimStudy-Resp10-Cov10-conjugate}-\ref{TABLE:LCRM:SimStudy-Resp0-Cov0-conjugate}. For Choice-II, we generate the same number of posterior samples with the same burn-in period as discussed before. Then, the posterior mean, SE, and RB are calculated based on every 10th iteration. In this case, we also repeat the whole procedure 500 times and report the average estimates of the parameters together with their CPs in Tables S1-S5 (see Section D in the Supplementary File). When $n = 50$, we observe that the biases are slightly higher for most of the parameters under Choice-II (Tables S1-S5) compared to Choice-I (Tables \ref{TABLE:LCRM:SimStudy-Resp10-Cov10-conjugate}-\ref{TABLE:LCRM:SimStudy-Resp0-Cov0-conjugate}). For instance, the RB for $\beta_{1S}$ is equal to 0.025 in Table \ref{TABLE:LCRM:SimStudy-Resp10-Cov10-conjugate}, while it is equal to 0.046 in Table S1. However, the CPs are approximately 95\% for most of the parameters under both the choices of hyperparameters. Thus, both the prior specifications provide similar results with respect to RB and CP in most of the cases.

\subsection{Computation time} \label{SEC:LCRM:Computational-Time}
We compare the runtime of Model-I, II, and III for each dataset in Table \ref{TABLE:LCRM:Run-Time} using the simulation settings indicated in Tables \ref{TABLE:LCRM:Model-Comparison-Conjugate}, \ref{TABLE:LCRM:Misspecified-Model-etaY=etaX=0.05}, and \ref{TABLE:LCRM:Misspecified-Model-etaY=etaX=0.075}. We run our codes in a computer with 8 GB RAM and Intel Xeon(R) CPU E5-2640 v2\@2.00GHz processor. The computation time increases as the proportion of zeros in Model-I and Model-III increases. Also, Model-I takes more computing time than Model-III, while Model-III takes longer to compute than Model-II. This is because of the generation of latent variables in the presence of zeros for both the longitudinal circular response and the circular covariate in Model-I, and only for the longitudinal circular response in Model-III. Same can be noted for the simulation settings that correspond to Tables \ref{TABLE:LCRM:SimStudy-Resp10-Cov10-conjugate}-\ref{TABLE:LCRM:SimStudy-Resp0-Cov0-conjugate}. Moreover, all the models take less than 3 minutes, even with considerable amount of zeros.
\begin{table}[!ht]
	\centering
	\caption{\textit{Comparison of average computing time in minutes}}
	\begin{tabular}{|l|l|l|l|}
		\hline
		~ & Model-I & Model-II & Model-III \\ \hline
		Table \ref{TABLE:LCRM:Model-Comparison-Conjugate} &\ 2.326 &\ 1.306 &\ 1.948 \\
		Table \ref{TABLE:LCRM:Misspecified-Model-etaY=etaX=0.05} &\ 2.406 &\ 1.367 &\ 2.196 \\
		Table \ref{TABLE:LCRM:Misspecified-Model-etaY=etaX=0.075} &\ 2.445 &\ 1.395 &\ 2.249 \\
		\hline
	\end{tabular}
	\label{TABLE:LCRM:Run-Time}
\end{table}

\section{Analysis of astigmatism data} \label{SEC:LCRM:Analysis-Astigmatism-Data}
We consider a dataset on cataract surgery conducted at the Disha Eye Hospital and Research Center, Barrackpore, West Bengal, India, from 2008 to 2010. A total of 56 patients were included in this study and treated with either SICS or PECS. They were subsequently followed up for three months, and the measurements on the axes of astigmatism were taken before the surgery and on the 1st, 7th, 30th, and 90th day post-surgery. The intensity of astigmatism for each of the patients is recorded before and after the 1st day of surgery, denoted by $I_0$ and $I_1$, respectively. In addition, the age and gender of the patients are also available. See \cite{Bakshi2010} for further details of the study. As discussed in Section \ref{SEC:LCRM:Introduction}, the main purpose of our analysis is to compare the recovery process of the patients treated with two different surgical procedures. It is also important for the medical practitioners to understand the recovery process of the patients based on their demographic profiles. 	
        
For this purpose, we first model the axis of astigmatism ($\theta_{Y}$) measured at the 7th, 30th, and 90th day after the surgery with the vector of linear covariates that include \textit{Age}, \textit{Gender}, \textit{Surgery}, and $I_1$. Here, \textit{Gender} and \textit{Surgery} are binary covariates represented by indicator variables. In this context, the females and males are coded as 1 and 0, respectively, and the surgery type PECS and SICS are indicated by 1 and 0, respectively. To understand the improvement of the patients over time, we further incorporate two dummy variables as covariates, denoted by $t_1$ and $t_2$, in Stage-I of our proposed model following \citet[Ch-4, pp. 56-75]{Twisk2023}. Note that $(t_1, t_2)$ takes values (0,0), (1,0), and (1,1) corresponding to 7th, 30th, and 90th day after the surgery, respectively. To understand the recovery process of the patients compared to the very 1st day of surgery, the axis of astigmatism on that day is also included as a circular covariate ($\theta_{X}$). Among the 56 individuals, there were $33.93\%$ zeros in the longitudinal circular response and $35.71\%$ zeros in the circular covariate. The dataset also contains $I_0$ and the measurement on the axis of astigmatism before surgery ($\theta_{V}$), which are considered as the linear and the circular instrumental variables, respectively. Since there are large proportion of zeros in the circular covariate, we model $\theta_{X}$ with covariates as $I_0$ and $\theta_{V}$ in Stage-II. Therefore, the aforementioned two-stage model is given by

\begin{align*}
     \text{Stage-I:} \quad Y_{1ij}^{*} = &\beta_{10} + \beta_{11} \textit{Age}_i + \beta_{12} \textit{Gender}_i + \beta_{13} \textit{Surgery}_i + \beta_{14} t_{1_{ij}} + \beta_{15} t_{2_{ij}} + \beta_{16} I_{1_i} \\ 
     &+ \beta_{1C} \cos{(\tXi^*)} + \beta_{1S} \sin{(\tXi^*)} + b_{1i} +  \varepsilon_{Y_{1ij}}, \\
     Y_{2ij}^{*} = &\beta_{20} + \beta_{21} \textit{Age}_i + \beta_{22} \textit{Gender}_i + \beta_{23} \textit{Surgery}_i + \beta_{24} t_{1_{ij}} + \beta_{25} t_{2_{ij}} + \beta_{26} I_{1_i} \\
     &+ \beta_{2C} \cos{(\tXi^*)} + \beta_{2S} \sin{(\tXi^*)} + b_{2i} + \varepsilon_{Y_{2ij}}, \\
     &\qquad \qquad \qquad \qquad \text{ and } \\
     \text{Stage-II:} \quad X_{1i}^{*} = &\alpha_{10} + \alpha_{11} I_{0_i} + \alpha_{1C} \cos{(\tVi)} + \alpha_{1S} \sin{(\tVi)} + \varepsilon_{X_{1i}}, \\
     X_{2i}^{*} = &\alpha_{20} + \alpha_{21} I_{0_i} + \alpha_{2C} \cos{(\tVi)} + \alpha_{2S} \sin{(\tVi)} + \varepsilon_{X_{2i}},
\end{align*}
for $i = 1, 2,\ldots, 56$ and $j = 1, 2, 3$.

For the analysis, we take $\delYtwo = \delXtwo = 0.035$ radians ($2^{\circ}$) since the original axes of astigmatism are censored within the interval (-$2^{\circ}, 2^{\circ}$) (see Section \ref{SEC:LCRM:Challenges-Circular-Regression}). We draw 0.7 million samples from the posterior distributions of the associated model parameters using the proposed Gibbs sampling algorithm in Section \ref{SEC:LCRM:Inferencial-Methodology}. Consequently, the posterior mean and SD of the parameters are calculated based on every 10th iterate discarding the first 0.15 million  iterations as burn-in. The convergence of the chains is monitored graphically using trace plots (see Figures S1 and S2 in Section F of the Supplementary File) and Geweke's diagnostic test (see Table S6 in Section E of the Supplementary File). The summary of the posterior estimates are provided in Table S6 in Section E of the Supplementary File.
\begin{figure}[ht!]
    \centering
    \begin{subfigure}{\linewidth}
        \centering
        \includegraphics[width=0.48\textwidth]{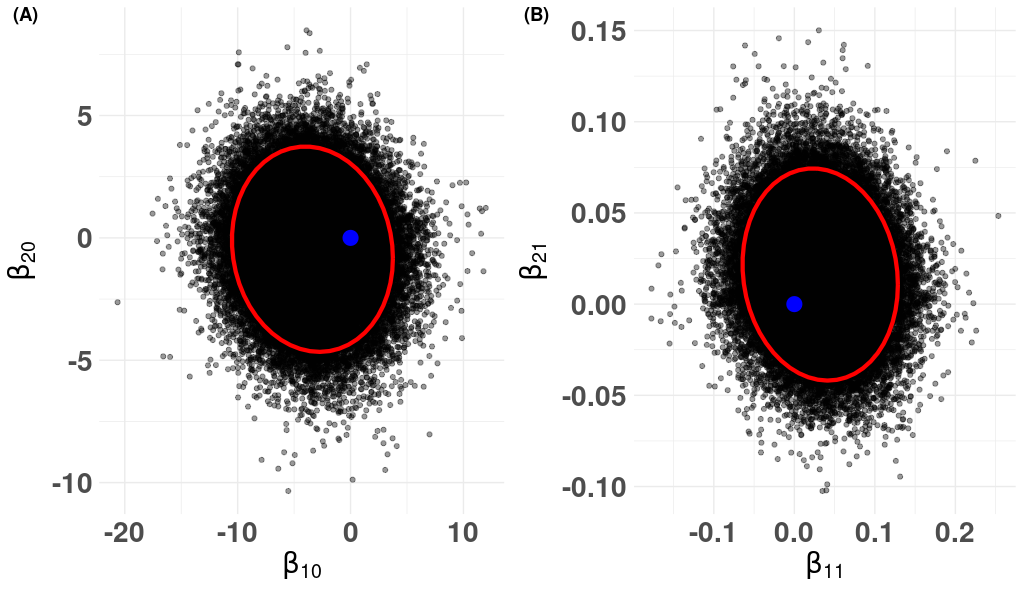}
        \includegraphics[width=0.48\textwidth]{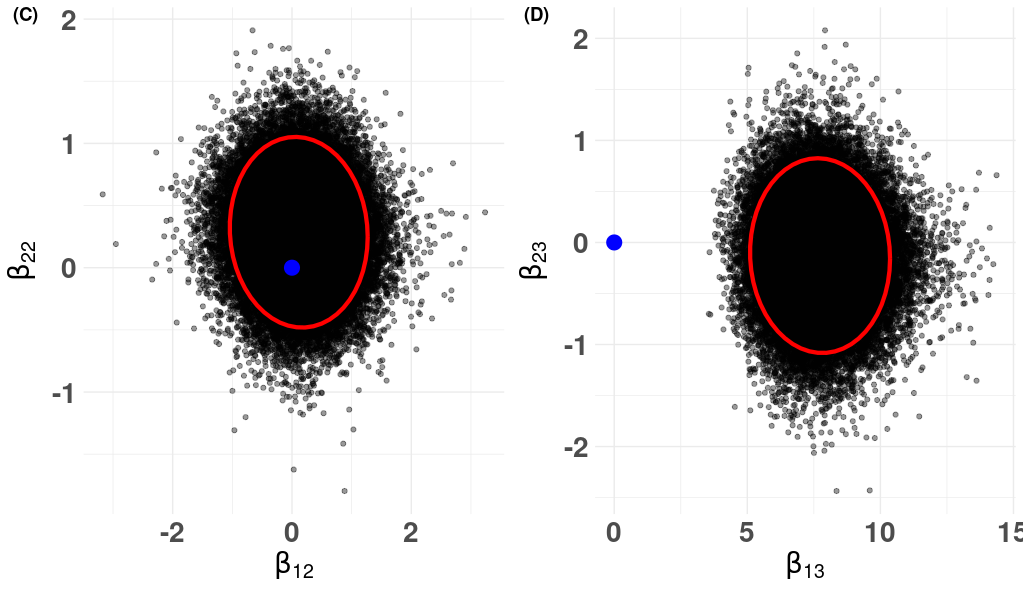}
    \end{subfigure}
    \begin{subfigure}{\linewidth}
        \centering
        \includegraphics[width=0.48\textwidth]{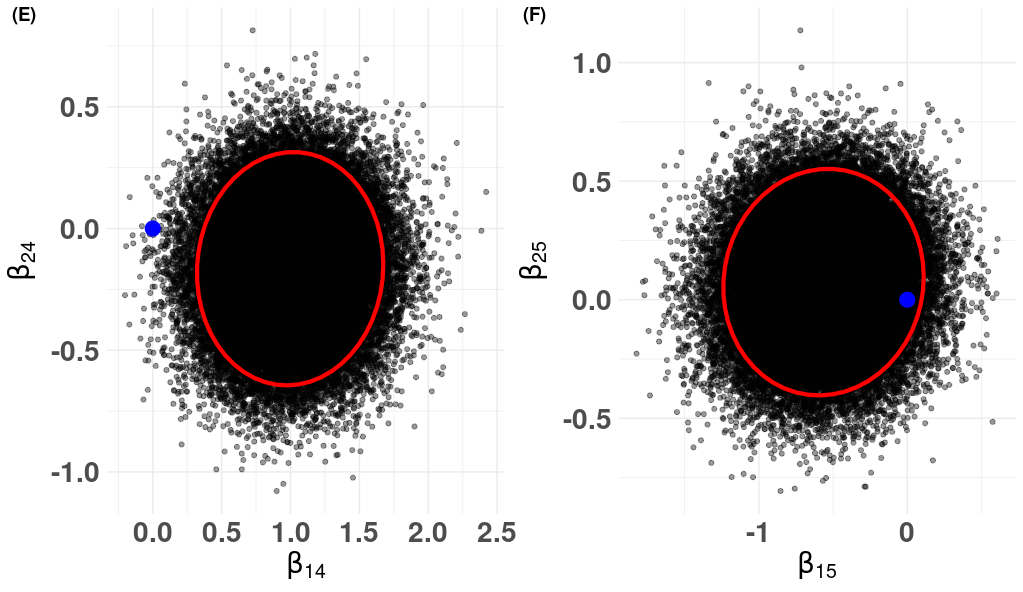}
        \includegraphics[width=0.48\textwidth]{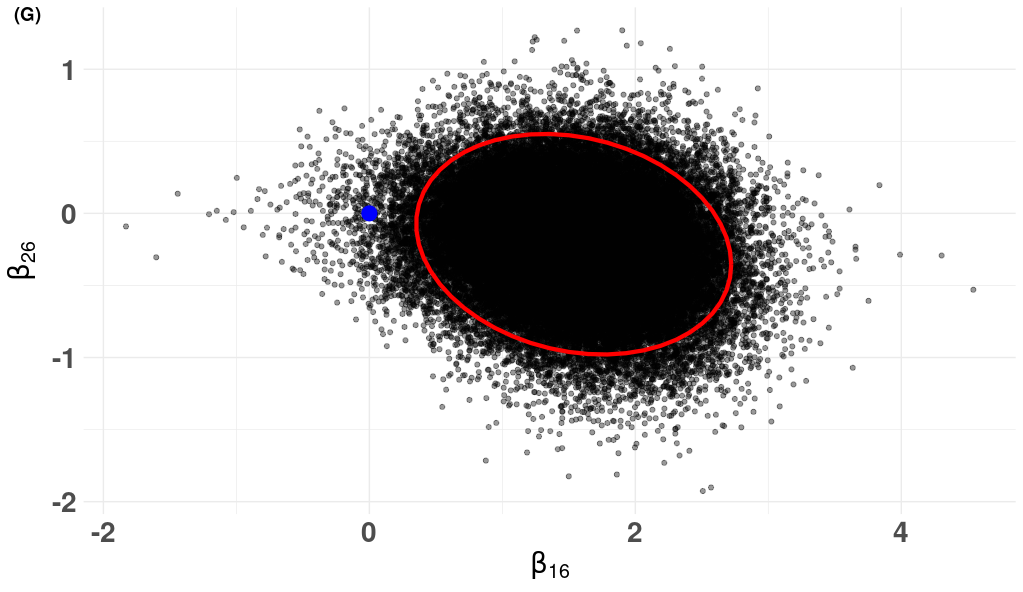}
    \end{subfigure}
    \caption{Confidence ellipses at 95\% level depicting the significance of (\textbf{A}) intercept, (\textbf{B}) \textit{Age}, (\textbf{C}) \textit{Gender}, (\textbf{D}) \textit{Surgery}, (\textbf{E}) $t_1$, (\textbf{F}) $t_2$, and (\textbf{G}) the intensity of astigmatism on the 1st day post-surgery at Stage-I, where `$\color{blue}{\bullet}$' indicates (0,0) position.}
    \label{FIG:LCRM:Significance-testing-plots-Stage-I}
\end{figure}

\begin{figure}[ht!]
    \begin{subfigure}{\linewidth}
        \centering
        \includegraphics[width=0.55\textwidth]{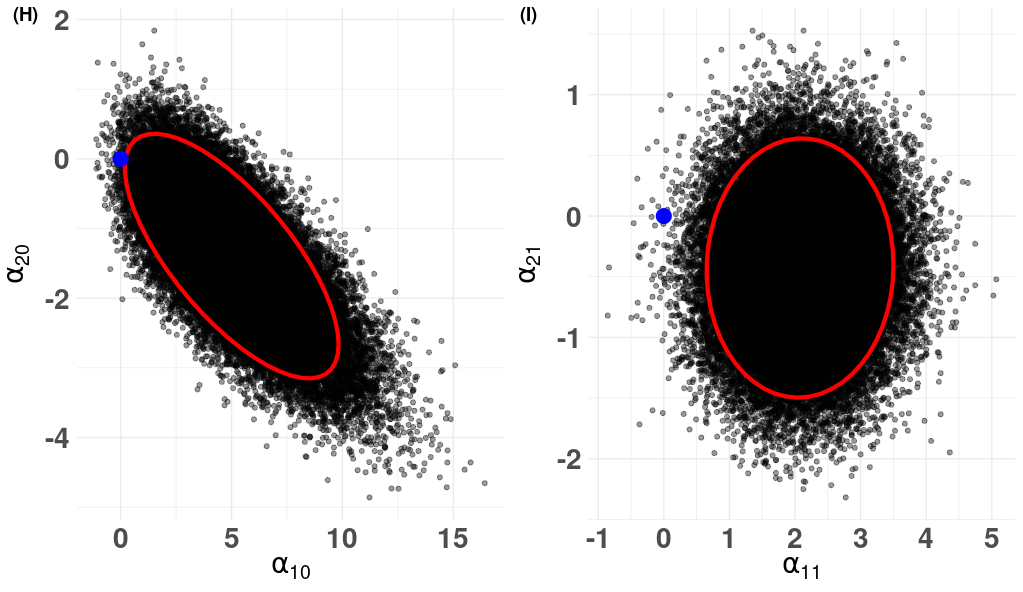}
    \end{subfigure}
    \caption{Confidence ellipses at 95\% level depicting the significance of (\textbf{H}) intercept and (\textbf{I}) the intensity of astigmatism before surgery at Stage-II, where `$\color{blue}{\bullet}$' indicates (0,0) position.}
    \label{FIG:LCRM:Significance-testing-plots-Stage-II}
\end{figure}

\begin{figure}
    \centering
    \begin{minipage}{0.49\textwidth}
        \centering
        \includegraphics[width=\linewidth]{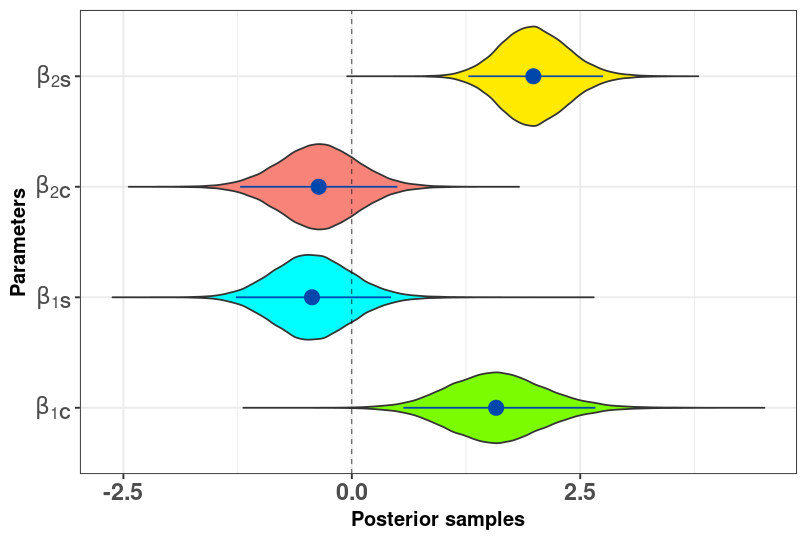}
        %\subcaption{$\theta_X$}
    \end{minipage}
    \medskip
    \begin{minipage}{0.49\textwidth}
        \centering
        \includegraphics[width=\linewidth]{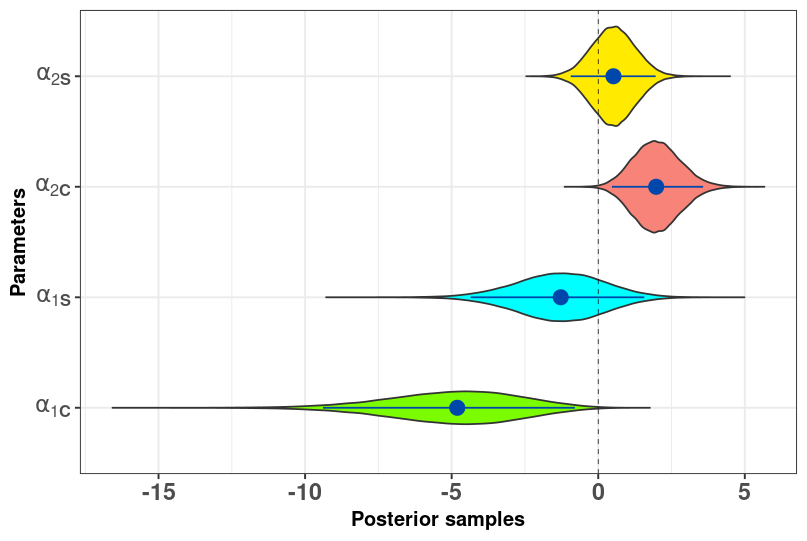}
        %\subcaption{$\theta_V$}
    \end{minipage}
    \caption{Violin plots depicting the significance of the axis of astigmatism on the 1st day post-surgery (left) and the axis of astigmatism before surgery (right), where `$\color{blue}{\bullet}$' and `$\color{blue}{\boldemdashbullet}$' indicate the median and the 95\% HPD credible interval of the posterior samples, respectively.}
    \label{FIG:LCRM:Significance-testing-thetaX-thetaV}
\end{figure}

The significance of the parameters have been checked graphically and shown in Figures \ref{FIG:LCRM:Significance-testing-plots-Stage-I} and \ref{FIG:LCRM:Significance-testing-plots-Stage-II}. These figures indicate that \textit{Surgery}, $t_1$ and $I_1$ are only significant linear covariates as their 95\% highest posterior density (HPD) credible intervals (CIs) do not contain (0,0) in Stage-I. Similarly, the intercept and the linear covariate $I_0$ both are significant in Stage-II. To test the significance of the circular covariate, we present the violin plots with 95\% HPD CIs in Figure \ref{FIG:LCRM:Significance-testing-thetaX-thetaV}, and check whether atleast one of the CIs of $\beta_{1C}, \beta_{1S}, \beta_{2C},$ and $\beta_{2S}$ do not contain zero. It is visible that the posterior distributions of $\beta_{1C}$ and $\beta_{2S}$ exhibit greater concentration far away from zero which indicates the significance of the axis of astigmatism at day 1. Similarly, the axis of astigmatism before surgery is a significant instrumental variable. To assess how well the proposed model fits the real data, we provide a donut-plot in Figure \ref{FIG:LCRM:Donut-Plot-Goodness-of-Fit}, proposed by \cite{Jha-Biswas2017}. Here, we plot the point $[1 + \cos{(\hat{\theta}_Y - \theta_Y)}](\cos{\hat{\theta}_Y}, \sin{\hat{\theta}_Y})$ by coloured dots, where $\theta_Y$ is the observed angle and $\hat{\theta}_Y$ is the predicted angle. In this figure, solid dots represent clockwise deviations, while hollow dots represent anticlockwise deviations of the predicted value from the corresponding observed value. This figure indicates a good fit, as more number of points are situated near the circumference of the wider circle centered at $0$ with a radius of $2$.
\begin{figure}[ht!]
    \centering
    \includegraphics[width=15cm]{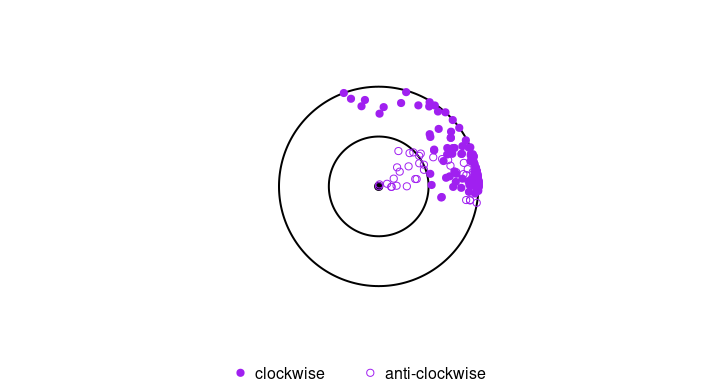}
    \caption{Donut-plot for astigmatism data.}
    \label{FIG:LCRM:Donut-Plot-Goodness-of-Fit}
\end{figure}

We refit the proposed model, discarding the insignificant covariates. We depict the recovery process of the patients after cataract surgery in Figure \ref{FIG:LCRM:Posterior-Predictive-Densities-SICS-PECS_along_time} for both surgical procedures. For comparison, the posterior predictive densities of the post-operative axis of astigmatism on the 7th, 30th, and 90th day are presented with four different initial conditions $0^\circ$ (normal), $45^\circ$ (moderate), $90^\circ$ (intermediate) and $180^\circ$ (serious), keeping $I_1 = 0.95$ which is the average value. For PECS, one can observe that the posterior predictive densities are more concentrated around $0^\circ$ irrespective of the initial conditions for almost all cases. On the other hand, the concentration of the posterior predictive densities for SICS is around $0^\circ$ only if the initial axis of astigmatism is close to $0^\circ$ on the very first day after surgery. We also observe that the variability in the posterior predictive distributions for PECS is lower compared to that of the SICS for all the  cases. This indicates greater uncertainty during the recovery process for those treated with SICS. In Figure \ref{FIG:LCRM:Posterior-Predictive-Densities-SICS-PECS_along_time}, it is visible that improvement is very slow over time and the changes are not much between consecutive occasions for SICS, when the patient's initial axis is $45^\circ$, $90^{\circ}$, or $180^\circ$. In particular, the improvement between the 7th and 30th day is marginally better for the patients with the initial condition $45^\circ$ compared to those with the initial conditions $90^\circ$ or $180^\circ$. Based on the posterior predictive densities, we have also checked that the average improvement of the axis of astigmatism on day 30 from that of the very 1st day is $8.34^\circ$, $19.48^\circ$, and $71.23^\circ$ under initial conditions $45^\circ$, $90^\circ$, and $180^\circ$, respectively. For all cases, a marginal deterioration is observed between the 30th and 90th days during the recovery process. See Table S7 in Section E of the Supplementary File for more details. This analysis can be helpful for clinicians to identify astigmatic patient whose performance is improving or deteriorating over time. We also performed the same analyses using the alternative priors given in Choice-II (see Section \ref{SEC:LCRM:Prior-Sensitivity}) and we found that the posterior predictive distributions were very similar to the ones obtained in Figure \ref{FIG:LCRM:Posterior-Predictive-Densities-SICS-PECS_along_time}. This indicates the robustness of our analysis with respect to different prior choices. 
\begin{figure}[ht!]
    \centering
    \includegraphics[width=0.95\textwidth]{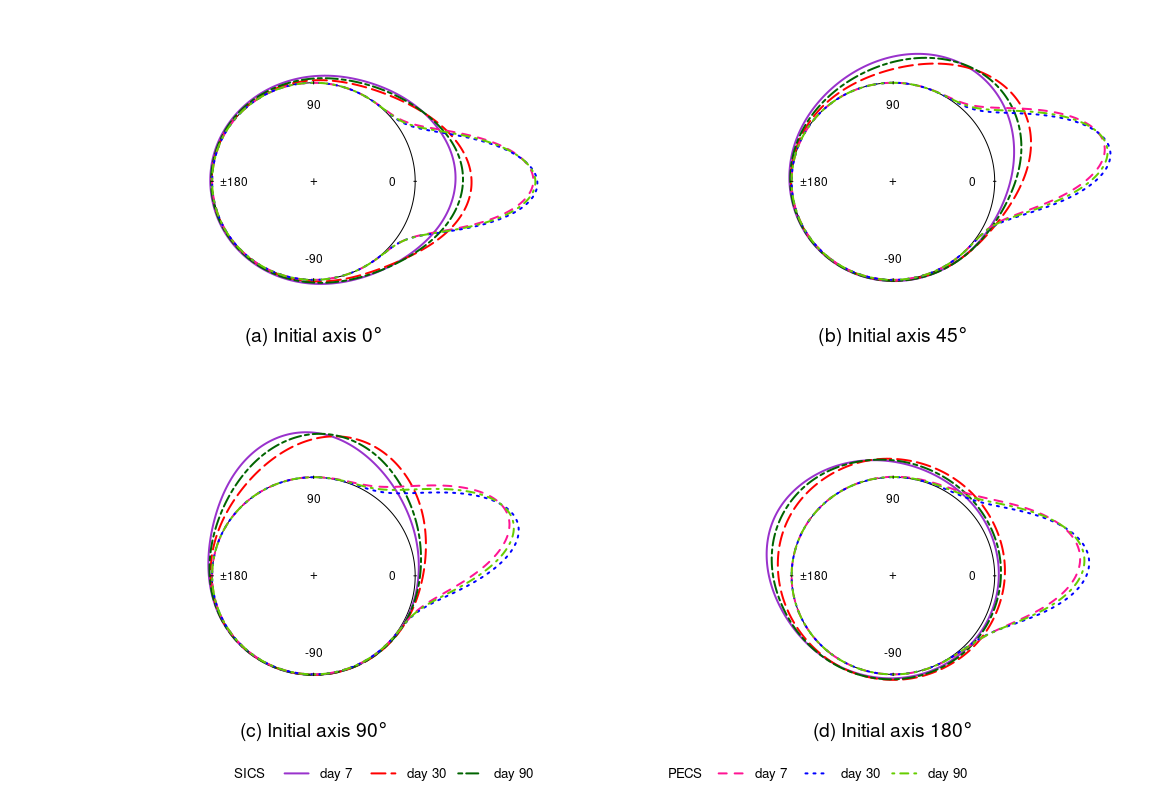}
    \caption{Comparison of the two surgeries through posterior predictive density plots based on four different initial conditions: (a) normal ($0^\circ$), (b) moderate ($45^\circ$), (c) intermediate ($90^\circ$) and (d) serious ($180^\circ$) case of the post-operative axis of astigmatism measured at the 7th, 30th, and 90th day.}
    \label{FIG:LCRM:Posterior-Predictive-Densities-SICS-PECS_along_time}
\end{figure}

\section{Discussions} \label{SEC:LCRM:Discussions}
In this paper, a longitudinal circular regression mixed-effects model has been proposed to account for the zero-inflation in both the circular response and the circular covariate. The proposed methodology is a step-forward in understanding the behaviour of the probability models for Euclidean spaces when projected on the unit sphere. Although the projected normal distributions were considered for circular data in the literature, a detailed understanding of the models was missing. This paper provides a thorough discussion in a general framework of longitudinal studies with a zero-inflated response as well as covariate. 
Interpretation of the model parameters and identifiability issues have been discussed in detail, and a Bayesian methodology has been developed for estimating the model parameters. The proposed method is applied to analyse postoperative astigmatism, which helps to understand the recovery process under different treatments. The generalization of the proposed methodology can cater to diverse scenarios arising from natural and physical phenomena. Therefore, the scope of the proposed method goes far beyond the particular case study under consideration. Although this paper does not consider missing data, one can readily extend the data augmentation strategy to sample the missing observations.

There are many applications, where the observations are longitudinal in nature, taking values on a unit sphere \citep{Wilson.etal2020}. To model such data, the proposed methodology can be extended by considering the latent space as $\mathbb{R}^d$, for $d>2$. However, this may require additional identifiability constraints, which can pose challenges in developing efficient computational algorithms for the estimation. Moreover, a projected normal distribution may not adequately model complex dependence structure among directional and/or linear variables. To address this issue, a non-Gaussian multivariate distribution on Euclidean space can be projected, where an efficient sampling algorithm is available. This is an interesting future problem from both theoretical and methodological perspectives.

\section*{Funding}
The work of Dr. Prajamitra Bhuyan is supported by the Category-I Research Project Grant (No. 3918/RP: BEMMS) from the Indian Institute of Management Calcutta, Kolkata, India. The work of Dr. Jayant Jha is partially funded by the Government of India as part of the Start-up Research Grant provided through the Science of Engineering Board of the Department of Science of Technology (No. SRG/2022/000151).

\section*{Acknowledgement}
The authors would like to thank Dr. Sourabh Bhattacharya for insightful comments and valuable suggestions, and acknowledge the support of Mr. Javed Hazarika and Mr. Sourojyoti Barick in \texttt{R} programming.

\section*{Author contributions}
All authors have contributed equally.

\section*{Additional information}
Supplementary Information: The online version contains \textbf{Supplement to ``Modeling Zero-Inflated Longitudinal Circular Data Using Bayesian Methods: Application to Ophthalmology"} available at \doi{10.13140/RG.2.2.17773.24807}.

\bibliographystyle{apalike} % unsrtnat.bst 
\bibliography{ref}

\end{document}